\PassOptionsToPackage{expansion=false}{microtype}
\documentclass{lmcs}

\usepackage{amsmath,amssymb,mathtools}
\usepackage{booktabs,array}

\makeatletter
\def\endfront@text{}
\def\enddoc@text{}
\makeatother

\newcommand{\Sym}{\mathrm{Sym}}
\newcommand{\Symfin}{\mathrm{Sym}_{\mathrm{fin}}}
\newcommand{\Fix}{\mathrm{Fix}}
\newcommand{\supp}{\mathrm{supp}}
\newcommand{\Inv}{\mathrm{Inv}}
\newcommand{\Aut}{\mathrm{Aut}}
\newcommand{\Rel}{\mathrm{Rel}}
\newcommand{\fs}{\mathrm{fs}}
\newcommand{\Nom}{\mathbf{Nom}}
\newcommand{\FSS}{\mathbf{FSS}}
\newcommand{\Set}{\mathbf{Set}}
\newcommand{\Sub}{\mathrm{Sub}}
\newcommand{\Sch}{\mathbf{Sch}}
\newcommand{\Pfs}{\mathcal{P}_{\fs}}
\newcommand{\Hom}{\mathrm{Hom}}
\newcommand{\Vfs}{V_{\fs}}
\newcommand{\FD}{\mathsf{FD}}

\begin{document}

\title[Finite Dependence and Invariance Hierarchies]
{Finite Dependence and Invariance Hierarchies for Finitely Supported Structures}

\author[Gabriel Ciobanu]{Gabriel Ciobanu\lmcsorcid{0000-0002-8166-9456}}
\address{Romanian Academy and A.I.Cuza University, Ia\c{s}i, Romania}
\email{gabriel.ciobanu@iit.academiaromana-is.ro}
\ACMCCS{Theory of computation~Automata over infinite objects;
Theory of computation~Categorical semantics;
Theory of computation~Logic and verification}

\begin{abstract}
Finite support indicates that the dependence is limited, but not how each finite context controls what can be observed. For a data symmetry $G\leq\Sym(A)$ and a finitely supported $G$-set $X$, we study the finite-dependence profile
$X^S=\{x:\Fix_G(S)\text{ fixes }x\}$, indexed by finite contexts
$S\subseteq A$. For relations on $A^n$ this yields Boolean algebras
$B^{(n)}_{G,S}=\Inv_n(\Fix_G(S))$, forming the invariance hierarchy.
The hierarchy separates several nominal principles. Its meet law is the
relation-algebra form of the Boja\'nczyk--Klin--Lasota least-support
criterion, while level injectivity is fungibility. The independent join
law governs composition across unions of contexts; it holds for locally
oligomorphic automorphism groups of ultrahomogeneous structures in purely
relational languages of arity at most two, but can fail for unary
observations over ternary data. Freshness remains valid sort by sort, and
its unrestricted some/any form characterizes full equality symmetry among
closed groups. With least supports, orbit-finiteness implies finite context
levels and a uniform support bound; under local oligomorphicity the converse
also holds.
Pointwise closure does not change the profile, and standard atomic-site
presentations show that suitable ultrahomogeneous structures with the same
age have equivalent action topoi, even though their material FSS universes
may retain different external cardinal data. Contextual orientability gives
a complementary pointed invariant whose threshold is the least support size
of choice from unordered pairs. For countable $\omega$-categorical
structures with degenerate algebraic closure, the meet law is equivalent to
weak elimination of imaginaries.
\end{abstract}

\maketitle

\section{Introduction}

Finite support is one of the most effective abstractions for computation
with names and infinite data. An element $x$ of a $G$-set is finitely
supported when there is a finite set $S$ of atoms such that every symmetry
fixing $S$ pointwise also fixes $x$. In classical nominal sets the atoms
carry only equality; in automata over atoms and related formalisms the
symmetry may preserve order, graph structure, equivalence classes, or other
relations. In the material theory of finitely supported structures (FSS),
the same support mechanism is used inside a set-theoretic universe with
atoms. These viewpoints share the same elementary support calculus, but
they emphasize different questions: nominal methods focus on syntax and
machines, data symmetries on structured infinite alphabets, and FSS on the
set-theoretic consequences of finite dependence.

The purpose of this paper is to isolate the common structural object behind
these uses of support. Rather than indicating only that an element has
\emph{some} finite support, we retain all finite contexts that support it.
For a finitely supported $G$-set $X$ and finite $S\subseteq A$ put
\[
   X^S=\{x\in X:\pi x=x\text{ for all }\pi\in\Fix_G(S)\}.
\]
Thus $X^S$ is the part of $X$ visible from the context $S$, and
$X=\bigcup_{S\in[A]^{<\omega}}X^S$. We call
$S\mapsto X^S$ the \emph{finite-dependence profile} of $X$. For the
powerset $G$-set of $A^n$ the profile becomes
\[
   B^{(n)}_{G,S}=\Inv_n(\Fix_G(S)),
\]
the Boolean algebra of $n$-ary relations supported by $S$. The family
$S\mapsto B^{(n)}_{G,S}$, over all finite contexts and arities, is the
\emph{invariance hierarchy}. It is the principal object studied here.

Finite support says that dependence is finite; the finite-dependence
profile shows how each finite context controls what can be observed. For a
fixed data symmetry $(A,G)$ this profile is canonically determined by the
action and preserved by equivariant isomorphism. We therefore treat it as a
structural invariant of the presented $G$-action. When invariance under an
arbitrary equivalence of action categories is intended, the additional
pointed or slice data will be stated explicitly. This distinction lets
intersection, union, freshness, finiteness, categorical asymmetry, and
material choice be compared in one language without being~conflated.

This point of view is compatible with, rather than opposed to, the general
nominal-$G$-set framework of Boja\'nczyk, Klin and Lasota
\cite{BKL14}. Their theory starts with an arbitrary data symmetry
$(D,G)$ and later isolates additional hypotheses, notably least supports and
fungibility, when finite representations are required. The present paper
does not re-prove those hypotheses under stronger assumptions. Instead it
places them inside the finite-context profile and then studies independent
phenomena carried by the same hierarchy. The material FSS viewpoint
\cite{AC20,AC22} contributes a second distinction: categorical finite-support semantics can make presentations with different
external cardinal data equivalent as abstract action topoi, whereas the
material universe retains the cardinality of the chosen atom sort.

\subsection{Contributions and provenance}\label{subsec:claims}

The results fall into five connected groups. We state the provenance here
because some parts of the support theory are established results that are
being reorganized rather than reproved as new criteria.

\begin{enumerate}[label=\textup{(\arabic*)},leftmargin=2.4em,itemsep=.45em]
\item \textbf{Context algebra: intersection and union.}
For finite contexts $S,T$, the meet identity
\[
 B^{(n)}_{G,S}\cap B^{(n)}_{G,T}=B^{(n)}_{G,S\cap T}
\]
is the relation-algebra form of the least-support criterion of
Boja\'nczyk--Klin--Lasota: their subgroup criterion and local orbit test hold for arbitrary data symmetries \cite[Def.~4.11, Fact~9.2,
Thm.~9.3]{BKL14}. We use them as calibration, and show additionally that
fungibility is exactly injectivity of the context-level map. By contrast,
the formally dual join identity
\[
 B^{(n)}_{G,S}\vee B^{(n)}_{G,T}=B^{(n)}_{G,S\cup T}
\]
is independent of least supports and asks whether every observation over a
combined context decomposes into observations over its two parts. It holds
for locally oligomorphic automorphism groups of ultrahomogeneous
structures in purely relational languages of arity at most two, while the generic $3$-uniform hypergraph and the betweenness
reduct exhibit genuinely mixed-context information.

\item \textbf{Freshness and state-space finiteness.}
The orbits of $\Fix_G(S)$ on $A\setminus S$ are the distinct freshness sorts
visible from $S$. A supported predicate is constant on each sort, so the
nominal \textsf{some/any} rule is always valid sort by sort; its unrestricted
one-sort form characterizes $\Sym(A)$ among closed groups. On arbitrary
state spaces the same profile separates two ingredients of orbit-finiteness:
finiteness of every level $X^S$ and a global bound on least-support size.
For symmetries with least supports, orbit-finiteness implies both
properties; under local oligomorphicity the converse holds as~well.

\item \textbf{Categorical asymmetry after naming parameters.}
For closed $G$, a context level has an intrinsic slice interpretation, so one
can ask which relational structures remain visible after the parameters in
$S$ are named. We develop orientability as one such invariant. Its
orbital/double-coset criterion isolates the swap obstruction, ordered
symmetries give positive examples, and contextual orientability measures how
many parameters must be named before the fresh atoms admit an invariant
tournament. The resulting threshold is exactly the least support size of a
choice function on unordered~pairs.

\item \textbf{Closure, cardinality, and the material/categorical split.}
The finite-dependence profile is unchanged when $G$ is replaced by its
pointwise closure, giving
$\FSS_G(A)\cong\FSS_{\overline G}(A)$ with the same support relation. The
standard atomic-site presentation then computes suitable closed action
categories from the age of an ultrahomogeneous presentation. In particular,
pure equality atoms of every infinite cardinality yield equivalent action
topoi. The material FSS universes remain distinguishable by their atom sort,
external cardinality, and support-sensitive choice phenomena. The
topos-theoretic presentation is standard
\cite{MacLaneMoerdijk,Caramello16}; the direct closure-invariance and material
comparison are included to fix precisely which information is structural and
which is presentation-dependent.

\item \textbf{Model-theoretic bridge.}
For countable $\omega$-categorical structures with degenerate algebraic
closure, the meet law is equivalent to weak elimination of imaginaries.
This identifies the least-support axis of the hierarchy with a standard
model-theoretic property while leaving the join axis independent. The
examples therefore separate support intersection, mixed-context dependence,
and categorical asymmetry rather than treating them as manifestations of one
undifferentiated notion of symmetry.
\end{enumerate}

The same finite-context profile supports these different interpretations, while the hypotheses and provenance of each remain explicit. Context composition and categorical orientability also admit more specialized developments; the present synthesis isolates the common finite-dependence mechanism and the boundaries between the resulting principles.

\subsection{Organization}
Section~\ref{sec:prelim} develops finite-dependence profiles and the
invariant-relation hierarchy. Sections~\ref{sec:regime}--\ref{sec:restrict} separate closure, cardinality and symmetry. Section~\ref{sec:meet}
calibrates the hierarchy against least supports and fungibility;
Section~\ref{sec:fresh} treats freshness and orbit-finiteness.
Section~\ref{sec:join} studies composition across unions of contexts, and Section~\ref{sec:orient} develops orientability and its contextual
threshold. Section~\ref{sec:modeltheory} gives the model-theoretic
correspondence, while Section~\ref{sec:modelling} collects the computational and foundational interpretations. The final section presents some open problems.

The paper therefore proceeds from a single support calculus to increasingly structured questions: first what finite contexts determine, then how context levels interact, and finally which categorical or material information survives changes of presentation.

\section{The support calculus} \label{sec:prelim}

Throughout, $A$ is an infinite set of \emph{atoms} and $G\leq\Sym(A)$ a subgroup, the \emph{symmetry parameter}. For finite $S\subseteq A$, $\Fix_G(S)=\{\pi\in G:\pi(a)=a\text{ for all }a\in S\}$ is the pointwise stabilizer. A $G$-set is a set $X$ with a left $G$-action; a finite $S\subseteq A$ \emph{supports} $x\in X$ if every
$\pi\in\Fix_G(S)$ fixes $x$, and $x$ is \emph{finitely supported} if some finite set supports it.

No purity assumption is imposed on the data represented by $A$. The set
$A$ is the underlying atom reservoir, while observable structure is encoded
by the chosen subgroup~$G$. In the principal structured examples one starts
from a relational structure $M$ with domain~$A$ and takes
$G=\Aut(M)$; for example $\Aut(\mathbb Q,<)$, the automorphism group of the
random graph, or the automorphism group of an equivalence relation. The
context lattice is always the ordinary lattice $[A]^{<\omega}$ of finite
subsets of the underlying atom set; changing the relational structure changes its stabilizers and hence the finite-dependence profile. Pure equality atoms are the special case $G=\Sym(A)$.

\begin{conv}\label{conv:standard}
Support is always understood in the standard, stabilizer-theoretic sense above. We do not use transitive-closure surrogates; in particular no special conventions are attached to small groups, and for the trivial group every element of every $G$-set is $\emptyset$-supported, so that $\FSS_{\{\mathrm{id}\}}(A)$ degenerates to the category of all sets: the symmetry parameter is informative only in the presence of symmetry.
\end{conv}

\noindent
We use ZFC as the surrounding metatheory. Thus ordinary set-indexed
choices used to construct external invariant relations are available.
No internal choice principle is assumed in either $\FSS_G(A)$ or the
material universe introduced in Section~\ref{sec:material}.

\begin{defi}\label{def:cat}
$\FSS_G(A)$ is the category whose objects are $G$-sets in which every
element is finitely supported, and whose morphisms are the
$G$-equivariant functions. We write $\Nom$ for
$\FSS_{\Symfin(\omega)}(\omega)$, equivalently (by the closure
invariance of Theorem~\ref{thm:closure} below, since
$\overline{\Symfin(\omega)}=\Sym(\omega)$) for
$\FSS_{\Sym(\omega)}(\omega)$; this is the category of nominal sets of \cite{Pitts13}.
\end{defi}

For a $G$-set $X$, write $\Pfs(X)$ for the set of subsets of $X$ that are
finitely supported under the image action. In particular,
\[
  \Pfs(A^n)
  =\bigcup_{S\in[A]^{<\omega}} B^{(n)}_{G,S},
  \qquad
  B^{(n)}_{G,S}=\bigl(\Pfs(A^n)\bigr)^S.
\]
Thus the invariance hierarchy is not an auxiliary collection of Boolean
algebras: it is exactly the finite-context filtration of the finitely
supported powerset familiar from the material FSS framework.

For a $G$-set $X$ and a finite context $S\subseteq A$, define the
\emph{$S$-level} of $X$ by
\[
   X^S:=\{x\in X: \pi x=x\text{ for every }\pi\in\Fix_G(S)\}.
\]
Thus $X^S$ is exactly the set of elements supported by $S$. When
$X\in\FSS_G(A)$, finite support is equivalently the exhaustion
\[
    X=\bigcup_{S\in[A]^{<\omega}}X^S.
\]
We call the monotone family $\FD_X:S\mapsto X^S$ the
\emph{finite-dependence profile} of $X$. It is not an additional
structure placed on $X$: it is recovered functorially from the $G$-action.
For $X=\mathcal P(A^n)$ the profile is the invariance hierarchy
$S\mapsto B^{(n)}_{G,S}$ introduced below. This terminology will let us
state freshness, finiteness and context-composition results at the same
level of abstraction.

The elementary facts of the support calculus are presented in the
generality of an arbitrary group action. We include their short proof
because later arguments use the support bounds explicitly.

\begin{lem}\label{lem:basic}
Let $X,Y$ be $G$-sets.
\begin{enumerate}[label=\textup{(\roman*)}]
\item If $S$ supports $x$ and $S\subseteq T$ with $T$ finite, then $T$ supports $x$.
\item If $S$ supports $x$ and $T$ supports $y$, then $S\cup T$
supports $(x,y)$.
\item If $f:X\to Y$ is equivariant and $S$ supports $x$, then $S$
supports $f(x)$.
\item A subset $U\subseteq X$ supported by $S$ is a union of
$\Fix_G(S)$-orbits.
\item $\pi\cdot x$ is supported by $\pi(S)$ whenever $x$ is supported
by $S$.
\end{enumerate}
\end{lem}
\begin{proof}
(i) Since $\Fix_G(T)\subseteq\Fix_G(S)$, every permutation fixing $T$
fixes $x$. For (ii), a permutation fixing $S\cup T$ fixes both
coordinates. For (iii), if $\pi\in\Fix_G(S)$, then
$\pi f(x)=f(\pi x)=f(x)$. For (iv), if $u\in U$ and
$\pi\in\Fix_G(S)$, support of $U$ gives $\pi u\in U$, so the entire
$\Fix_G(S)$-orbit of $u$ lies in $U$. Finally, if
$\rho\in\Fix_G(\pi S)$, then $\pi^{-1}\rho\pi\in\Fix_G(S)$ and
$\rho(\pi x)=\pi(\pi^{-1}\rho\pi)x=\pi x$, which proves (v).
\end{proof}

In profile notation, Lemma~\ref{lem:basic} says in particular that
$S\subseteq T$ implies $X^S\subseteq X^T$, that an equivariant map
$f:X\to Y$ sends $X^S$ into $Y^S$, and that
$\pi(X^S)=X^{\pi S}$. These elementary compatibilities are the reason
finite contexts can be treated as a lattice of observable dependence.

\begin{lem}\label{lem:subobj}
For $X\in\FSS_G(A)$, the subobjects of $X$ are (represented by) the
$G$-invariant subsets of $X$, ordered by inclusion. In particular, an
\emph{atom} of $\FSS_G(A)$ (an object with exactly two subobjects) is the same thing as a single nonempty $G$-orbit of finitely
supported~elements.
\end{lem}

\begin{proof}
An equivariant injection $Z\to X$ has $G$-invariant image, and its
elements are finitely supported in $X$; conversely a $G$-invariant
subset of $X$ is an object of $\FSS_G(A)$ and includes equivariantly.
Two monomorphisms represent the same subobject iff they have the same~image.
\end{proof}

For $n\geq1$, $\Rel^{(n)}_{\fs}(A)$ is the set of finitely
supported subsets of $A^n$ (with the coordinatewise action), and for
finite $S\subseteq A$, $B^{(n)}_{G,S} = \{R\subseteq A^n:\ S\text{ supports }R\} = \Inv_n\bigl(\Fix_G(S)\bigr)$ is a Boolean subalgebra of $\mathcal P(A^n)$; the displayed identity is
immediate from Lemma~\ref{lem:basic}(iv). For a family $\mathcal R$
of $n$-ary relations, $\Aut_G(\mathcal R)=\{\pi\in G:\pi\cdot R=R\text{ for all } R\in\mathcal R\}$. The following two lemmas indicate the elementary interface between subgroups and supported relation algebras that underlies this interpretation. Both are standard consequences of the subgroup--relation Galois connection
(short arguments are included to keep the paper self-contained).

\begin{lem}[Galois connection and stabilizer recovery]
\label{lem:galois}
Fix $n\geq1$.
\begin{enumerate}[label=\textup{(\roman*)}]
\item The assignments $H\mapsto\Inv_n(H)$ and
$\mathcal R\mapsto\Aut_G(\mathcal R)$ form an antitone Galois
connection between subgroups of $G$ and families of $n$-ary relations: $\mathcal R\subseteq\Inv_n(H)$ if and only if
$H\leq\Aut_G(\mathcal R)$.
\item For every finite $S\subseteq A$,
$\Aut_G\bigl(B^{(n)}_{G,S}\bigr)=\Fix_G(S)$.
\end{enumerate}
Consequently, the level map $\ell_n:S\mapsto B^{(n)}_{G,S}$ is
monotone, and it is injective whenever $S\mapsto\Fix_G(S)$ is
injective.
\end{lem}
\begin{proof}
For (i), both sides assert that every element of $H$ fixes every relation in $\mathcal R$ setwise. For (ii), the inclusion $\supseteq$ is the
definition of $B^{(n)}_{G,S}$. For $\subseteq$, fix $s\in S$ and 
$R_s:=\{\bar a\in A^n: a_1=s\}$. Every $\pi\in\Fix_G(S)$ fixes
$R_s$, so $R_s\in B^{(n)}_{G,S}$; and for arbitrary $\pi\in G$ one has
$\pi\cdot R_s=R_{\pi(s)}$, which equals $R_s$ only if $\pi(s)=s$.
Hence any $\pi$ preserving every member of $B^{(n)}_{G,S}$ fixes $S$
pointwise. Monotonicity of $\ell_n$ follows from antitonicity of
$\Inv_n$; and if $B^{(n)}_{G,S}=B^{(n)}_{G,T}$, then applying (ii) gives $\Fix_G(S)=\Fix_G(T)$. Thus equality of
levels forces equality of pointwise stabilizers; whenever the map
$S\mapsto\Fix_G(S)$ is injective, it also forces $S=T$.
\end{proof}

\begin{lem}[Level injectivity and fungibility]\label{lem:fungible}
Let $G\leq\Sym(A)$. The following are equivalent.
\begin{enumerate}[label=\textup{(\roman*)}]
\item The map $S\mapsto\Fix_G(S)$ is injective on finite subsets of $A$.
\item The level map $\ell_n:S\mapsto B^{(n)}_{G,S}$ is injective for some, equivalently for every, $n\geq1$.
\item Every finite $C\subseteq A$ is fungible in the sense of \cite[Def.~9.6]{BKL14}: for each $c\in C$ there is $\pi\in G$ such that $\pi(c)\neq c$ and $\pi$ fixes $C\setminus\{c\}$ pointwise.
\end{enumerate}
\end{lem}
\begin{proof}
If $\Fix_G(S)=\Fix_G(T)$, then $B^{(n)}_{G,S}=B^{(n)}_{G,T}$ for every $n$. Conversely, equality of two levels at any fixed arity implies equality of the corresponding pointwise stabilizers by Lemma~\ref{lem:galois}(ii). This proves (i)$\Leftrightarrow$(ii).

If a finite $C$ is not fungible, choose $c\in C$ such that every element of $\Fix_G(C\setminus\{c\})$ fixes~$c$. Then $\Fix_G(C\setminus\{c\})=\Fix_G(C)$, so (i) fails. Thus (i) implies (iii). Conversely suppose (iii) and $\Fix_G(S)=\Fix_G(T)$. If $S\neq T$, interchange $S,T$ if necessary and choose $c\in T\setminus S$. Put $U=S\cup T$ and $V=U\setminus\{c\}$. Since $S\subseteq V\subseteq U$,
\[
  \Fix_G(U)\leq\Fix_G(V)\leq\Fix_G(S).
\]
But $\Fix_G(U)=\Fix_G(S)\cap\Fix_G(T)=\Fix_G(S)$, so all three groups are equal. Fungibility of~$U$ at~$c$, however, supplies an element of $\Fix_G(V)\setminus\Fix_G(U)$, a contradiction. Hence $S=T$, proving (i).
\end{proof}

Thus the injectivity condition on context levels is not an auxiliary regularity assumption: it is exactly the fungibility condition isolated in \cite{BKL14}. In the standard homogeneous examples below it can be checked directly from the extension property; when it fails, distinct finite contexts collapse to the same level of the hierarchy.

\begin{lem}[Finite atomicity]\label{lem:atomic}
If $\Fix_G(S)$ has finitely many orbits on $A^n$, then
$B^{(n)}_{G,S}$ is a finite atomic Boolean algebra whose atoms are
exactly those orbits.
\end{lem}

\begin{proof}
By Lemma~\ref{lem:basic}(iv), applied to the coordinatewise action on
$A^n$, a relation is $S$-supported if and only if it is a union of
$\Fix_G(S)$-orbits on $A^n$; with finitely many orbits,
$B^{(n)}_{G,S}$ is therefore the full powerset of the (finite) set of
orbits.
\end{proof}

By Lemma~\ref{lem:fungible}, injectivity of the level map may equivalently be read as fungibility of the data symmetry. We invoke it below only in examples where fungibility is immediate from the available extension or homogeneity property; the BKL examples presented in Proposition~\ref{prop:bklpathologies} show that it is independent of least supports.

Finally, we say $G$ is \emph{locally oligomorphic} if $\Fix_G(S)$ has
finitely many orbits on $A^n$ for every finite $S$ and every $n$; under this hypothesis each $B^{(n)}_{G,S}$ is, by Lemma~\ref{lem:atomic}, a finite atomic Boolean algebra whose atoms are the $\Fix_G(S)$-orbits on $A^n$. We note that $\Sym(A)$ is locally oligomorphic for \emph{every} infinite $A$: the orbits of
$\Fix_{\Sym(A)}(S)$ on $A^n$ are classified by equality types with
parameters in $S$, a finite list independent of $|A|$. This
cardinal-independence of the orbit calculus is the combinatorial germ
of the cardinality-invariance result below.

\section{Closure invariance and the age presentation}
\label{sec:regime}

\emph{Logical dependencies.} The meet-law characterization of
Section~\ref{sec:meet} does not depend on this section or on
Section~\ref{sec:material}; a reader whose interest is that theorem may
proceed directly. What is established here is used in the categorical
comparison results of Section~\ref{sec:orient} and in fixing the scope of
the model-selection criteria of Section~\ref{sec:modelling}.

This section establishes the cardinality and topology baseline on which
the later comparison of nominal principles rests. The route is
modular: a finitary-dependence lemma valid for every symmetry
parameter; a closure-invariance theorem reducing every $\FSS_G(A)$ and its finite-context observable hierarchy to
the case of a closed group; a self-contained presentation of the
closed case as sheaves on a site of tuple orbits; and the computation
of that site from the age of an ultrahomogeneous structure. Cardinality
invariance for pure sets follows as a corollary, with a proof independent of
the cardinality of $A$ at every step.

\subsection{Finitary dependence and closure invariance}

\begin{lem}[Finitary dependence]\label{lem:dependence}
Let $G\leq\Sym(A)$, let $X$ be a $G$-set, and let $x\in X$ be
supported by the finite set $S$. If $\pi,\pi'\in G$ agree on $S$,
then $\pi\cdot x=\pi'\cdot x$.
\end{lem}

\begin{proof}
$\pi'\pi^{-1}\in G$ fixes $\pi(S)$ pointwise: for $s\in S$,
$(\pi'\pi^{-1})(\pi s)=\pi'(s)=\pi(s)$. By
Lemma~\ref{lem:basic}(v), $\pi(S)$ supports $\pi\cdot x$; thus
$\pi'\cdot x=(\pi'\pi^{-1})\cdot(\pi\cdot x)=\pi\cdot x$.
\end{proof}

Recall that $\Sym(A)$ carries the topology of pointwise convergence,
with the pointwise stabilizers $\Fix_{\Sym(A)}(S)$, $S$ finite, a
neighbourhood basis of the identity; for $G\leq\Sym(A)$,~$\overline G$ denotes its closure, such that $\bar\pi\in\overline G$ iff for every finite $F\subseteq A$ some $\pi\in G$ agrees with $\bar\pi$ on~$F$.

\begin{thm}[Closure invariance]\label{thm:closure}
For every $G\leq\Sym(A)$ there is an isomorphism of categories
$\FSS_{\overline G}(A)\cong\FSS_G(A)$
which is the identity on underlying sets and functions. A finite set
supports an element for the $G$-action if and only if it supports it for
the extended $\overline G$-action. Hence every category in the
family has a canonical presentation by a closed symmetry parameter.
\end{thm}
\begin{proof}
Restriction along $G\leq\overline G$ sends every finitely supported
$\overline G$-set and every equivariant map to the corresponding
$G$-object and $G$-map. It remains to extend a finitely supported
$G$-action uniquely to $\overline G$.

Let $X$ be a finitely supported $G$-set. If $\bar\pi\in\overline G$
and $x\in X$ has finite support $S$, choose $\pi\in G$ agreeing with
$\bar\pi$ on $S$ and define $\bar\pi\cdot x:=\pi\cdot x$.
If $\pi'$ is another such approximation, Lemma~\ref{lem:dependence}
gives $\pi x=\pi'x$. If $T$ is another support, choose an approximation
on $S\cup T$ and apply the same lemma twice; thus the value is also
independent of the chosen support.

The identity acts trivially by choosing the identity approximation.
For composition, let~$S$ support $x$, choose $\sigma\in G$ agreeing
with $\bar\sigma$ on $S$, and choose $\pi\in G$ agreeing with
$\bar\pi$ on $\sigma(S)=\bar\sigma(S)$. Then
$\bar\pi(\bar\sigma x)=\pi(\sigma x)=(\pi\sigma)x$,
and $\pi\sigma$ agrees with $\bar\pi\bar\sigma$ on $S$; hence the
right-hand side is $(\bar\pi\bar\sigma)x$. This proves the action
law, and the action extends the original one because an element of
$G$ may be chosen as its own approximation.

The extension is unique among finitely supported $\overline G$-actions
extending the given $G$-action. Indeed, for any such extension let
$T$ support $x$, and choose $\pi\in G$ agreeing with $\bar\pi$ on
$S\cup T$. Then $\pi^{-1}\bar\pi$ fixes $T$ pointwise, so the other
extension sends $x$ to $\pi x$; because $\pi$ also agrees with
$\bar\pi$ on $S$, this is the value constructed above.

If $S$ supports $x$ for $G$ and
$\bar\pi\in\Fix_{\overline G}(S)$, choose the approximation
$\pi\in\Fix_G(S)$; then $\bar\pi x=\pi x=x$. Conversely, support for
$\overline G$ implies support for the subgroup $G$. Finally, if
$f:X\to Y$ is $G$-equivariant and $S$ supports $x$, choose
$\pi\in G$ agreeing with $\bar\pi$ on $S$. By
Lemma~\ref{lem:basic}(iii), $S$ also supports $f(x)$, and thus
$f(\bar\pi x)=f(\pi x)=\pi f(x)=\bar\pi f(x)$.
Therefore the extension is functorial, and restriction and extension are
mutually inverse.
\end{proof}

\begin{cor}[Closure invariance of finite-context observables]
\label{cor:closure-hierarchy}
For every finite $S\subseteq A$ and every $n\geq1$,
\[
   B^{(n)}_{G,S}=B^{(n)}_{\overline G,S}.
\]
Consequently $G$ and $\overline G$ have the same finite-tuple orbits,
and local oligomorphicity, the meet law, the join law, and the
least-support property depend only on the closure of the symmetry
parameter.
\end{cor}
\begin{proof}
The closure of $\Fix_G(S)$ in $\Sym(A)$ is
$\Fix_{\overline G}(S)$: if $\bar g$ belongs to the latter and
$F\subseteq A$ is finite, approximate $\bar g$ on $F\cup S$ by an
element of $G$, which then fixes $S$ pointwise. If
$R\in B^{(n)}_{G,S}$ and $\bar g\in\Fix_{\overline G}(S)$, then
for each tuple $\bar a\in A^n$ choose such an approximation agreeing
with $\bar g$ on $S$ and on the entries of $\bar a$. Invariance of
$R$ under the approximation gives
$\bar a\in R$ if and only if $\bar g\bar a\in R$. Thus
$R\in B^{(n)}_{\overline G,S}$; the reverse inclusion is immediate.
The same approximation argument gives equality of the finite-tuple
orbit partitions. The remaining properties are expressed entirely
through these orbit partitions, the algebras $B^{(n)}_{G,S}$, or the
support relation, which is identical by Theorem~\ref{thm:closure}.
\end{proof}

\begin{cor}\label{cor:cont}
For closed $G\leq\Sym(A)$, an action of $G$ on a discrete set $X$ is
continuous if and only if every element of $X$ is finitely supported;
thus $\FSS_G(A)$ is the category $\mathrm{Cont}(G)$ of continuous
discrete $G$-sets, and, for arbitrary $G$,
$\FSS_G(A)\cong\mathrm{Cont}(\overline G)$.
\end{cor}
\begin{proof}
Continuity of $G\times X\to X$ at $(e,x)$, $X$ discrete, means the
stabilizer of $x$ is open, i.e.\ contains some $\Fix_G(S)$, i.e.\ $x$
is finitely supported; continuity everywhere follows by translation.
The rest is Theorem~\ref{thm:closure}.
\end{proof}

\subsection{The tuple-orbit site}

Fix a closed $G\leq\Sym(A)$. For a finite tuple
$\bar a=(a_1,\dots,a_k)$ of atoms ($k\geq0$) let
$O_{\bar a}:=G\cdot\bar a\subseteq A^k$ denote its orbit, an object of
$\FSS_G(A)$ with the coordinatewise action; $O_{()}$ is the terminal
object. Let $\mathbb T_G$ be the full subcategory of $\FSS_G(A)$ on
the objects $O_{\bar a}$; it is the \emph{orbit category} of the
transitive continuous $G$-sets of the form $G/\Fix_G(\bar a)$. Every
morphism between transitive $G$-sets is a surjection determined by its value at one base point, and $O_{\bar d}\to O_{\bar c}$ exists (sending $\bar d\mapsto\lambda\bar c$) exactly when
$\Fix_G(\bar d)\subseteq\Fix_G(\lambda\bar c)$. The distinguished
\emph{selections} $\sigma:O_{\bar d}\to O_{\bar c}$,
$\pi\bar d\mapsto\pi\bar c$, defined whenever the entries of $\bar c$
occur among those of $\bar d$, will index the age computation below.

\begin{lem}[Ore condition]\label{lem:ore}
Every cospan $f:O_{\bar c}\to O_{\bar b}\leftarrow O_{\bar e}:g$ in
$\mathbb T_G$ can be completed to~a commutative square from some
$O_{\bar p}$. Consequently the \emph{atomic topology}
$J_{\mathrm{at}}$, whose covering sieves are the nonempty sieves, is a Grothendieck topology on $\mathbb T_G$.
\end{lem}
\begin{proof}
The pullback $P=\{(w,w'):f(w)=g(w')\}$ in $G$-sets is nonempty because $f,g$ are surjective. Pick $(w,w')\in P$ and a tuple $\bar p$ whose
entries support both $w$ and $w'$ (Lemma~\ref{lem:basic}(ii)); then $\pi\bar p\mapsto(\pi w,\pi w')$ is a well-defined morphism $O_{\bar p}\to P$ by Lemma~\ref{lem:dependence}, and composing with the projections completes the square. The Ore condition makes nonempty sieves stable under pullback: if a sieve on $O_{\bar b}$ contains an arrow $f$ and $g$ is any arrow into $O_{\bar b}$, an Ore square provides an arrow in the pullback sieve. The maximal sieve is nonempty. Finally, if a sieve $R$ is locally nonempty along a nonempty covering sieve, choose one covering arrow $f$ and one arrow in $f^*R$; their composite belongs to $R$, so $R$ is nonempty. These verifications give the Grothendieck-topology axioms \cite{MacLaneMoerdijk}.
\end{proof}

The site $(\mathbb T_G,J_{\mathrm{at}})$ is atomic, so its sheaves may be tested by the single-morphism descent criterion
\cite{MacLaneMoerdijk}; we use this only through the standard
presentation theorem invoked~below.

\begin{thm}[Site presentation]\label{thm:site}
For closed $G\leq\Sym(A)$, the functor
\[
\Phi:\FSS_G(A)\longrightarrow
\mathrm{Sh}(\mathbb T_G,J_{\mathrm{at}}),\qquad
\Phi(X)=\Hom_{\FSS_G(A)}(-,X)\restriction\mathbb T_G,
\]
is an equivalence of categories.
\end{thm}

\begin{proof}
By Corollary~\ref{cor:cont}, $\FSS_G(A)=\mathrm{Cont}(G)$ is the topos of continuous discrete $G$-sets. The result is then the classical presentation of $\mathrm{Cont}(G)$ as an atomic sheaf topos; we recall the argument for completeness, and refer to \cite[III.9]{MacLaneMoerdijk} and, for the version closest to the form used below, to Caramello's topological Galois theory \cite[Prop.~2.3 and Cor.~3.7]{Caramello16}. For a topological group $G$ whose open subgroups form a neighbourhood basis of the identity (here the pointwise stabilizers $\Fix_G(S)$, $S$ finite) the
category $\mathrm{Cont}(G)$ is a Grothendieck topos, and it is presented by its full subcategory of \emph{transitive} continuous objects equipped with the atomic topology: writing $\mathcal O_G$ for the category whose objects are the transitive continuous $G$-sets $G/U$ ($U$ open) and whose morphisms are the equivariant maps,
\[
\mathrm{Cont}(G)\;\simeq\;
\mathrm{Sh}\bigl(\mathcal O_G,\,J_{\mathrm{at}}\bigr),
\qquad X\longmapsto\Hom(-,X)\restriction\mathcal O_G,
\]
by \cite[III.9, Thm.~1]{MacLaneMoerdijk}. No separability, second-countability, metrizability, or Polishness hypothesis is used in that theorem: it is stated for an arbitrary topological group. Thus,
\[
   \mathrm{Cont}(G)
   \simeq
   \mathrm{Sh}\bigl(\mathcal O_G,J_{\mathrm{at}}\bigr),
\]
where $\mathcal O_G$ is the category of transitive continuous $G$-sets $G/U$, with $U$ ranging over the open subgroups of $G$, and $J_{\mathrm{at}}$ is the atomic topology.

It remains to replace the full orbit category $\mathcal O_G$ by the finite-tuple orbit category $\mathbb T_G$. The inclusion
$\mathbb T_G\hookrightarrow\mathcal O_G$ is full and faithful. Moreover, every object of $\mathcal O_G$ is covered by an object of $\mathbb T_G$.  Indeed, if $G/U$ is a transitive continuous $G$-set, then $U$ is open. Since the pointwise stabilizers of finite tuples form a neighbourhood basis of the identity, there is a finite tuple $\bar a$ such that $\Fix_G(\bar a)\leq U$. Consequently the canonical map
\[
   O_{\bar a}
   \cong G/\Fix_G(\bar a)
   \longrightarrow G/U,
   \qquad
   g\Fix_G(\bar a)\longmapsto gU,
\]
is a well-defined surjective $G$-map, hence a cover for the atomic topology.

Therefore $\mathbb T_G$ is $J_{\mathrm{at}}$-dense in $\mathcal O_G$. By the Comparison Lemma \cite[VII.10, Cor.~3]{MacLaneMoerdijk}, restriction along $\mathbb T_G\hookrightarrow\mathcal O_G$ induces an equivalence
\[
   \mathrm{Sh}\bigl(\mathcal O_G,J_{\mathrm{at}}\bigr)
   \simeq
   \mathrm{Sh}\bigl(\mathbb T_G,J_{\mathrm{at}}\bigr).
\]
Combining the two equivalences gives
\[
   \FSS_G(A)
   =\mathrm{Cont}(G)
   \simeq
   \mathrm{Sh}\bigl(\mathbb T_G,J_{\mathrm{at}}\bigr).
\]
Under this composite equivalence, an object $X$ is sent to the restriction of the representable presheaf $\Hom_{\FSS_G(A)}(-,X)$ to $\mathbb T_G$, which is precisely the functor $\Phi$ appearing in the statement.
\end{proof}

\begin{rem}[Scope of the external topos-theoretic input] \label{rem:selfcont}
The only non-elementary input in Theorem~\ref{thm:site} is the classical atomic comparison theorem presenting $\mathrm{Cont}(G)$ by its transitive continuous actions. The form used here is \cite[III.9, Thm.~1]{MacLaneMoerdijk}, which is stated for an \emph{arbitrary topological group} $G$ (no separability, second-countability, metrizability, or Polishness hypothesis is imposed). The variant \cite[III.9, Thm.~2]{MacLaneMoerdijk} likewise applies to an arbitrary cofinal system of open subgroups.

Consequently the comparison theorem applies without change to $\Sym(A)$ with the topology of pointwise convergence when $A$ is uncountable.  Although this group is then generally non-Polish and non-second-countable, the finite point stabilizers $\Fix_{\Sym(A)}(S)$, $S\in[A]^{<\omega}$, still form a neighbourhood basis of the identity and a cofinal system among the open subgroups in the sense required above.

All statements specific to finite support are proved in the present paper: closure invariance, the identification with continuous discrete
actions, the cofinality of the finite-tuple orbits, and the Ore condition. We do not reproduce the general comparison theorem, since
doing so would amount to redeveloping a standard part of atomic-topos theory; the verification above indicates every hypothesis needed for its application. In particular, countability enters only in later model-theoretic results, such as the $\omega$-categorical weak-elimination correspondence, and not in the closure, continuous-action, or atomic-site results of this section.
\end{rem}

\subsection{Ultrahomogeneous structures and age invariance}

\begin{defi}\label{def:generic}
A relational structure $M$ (of any infinite cardinality) is
\emph{ultrahomogeneous} if every isomorphism between finite
substructures extends to an automorphism. Its \emph{age}
$\mathcal A(M)$ is the class of finite structures embeddable in $M$.
We say $M$ has \emph{trivial stabilizer-dcl} if for every finite tuple
$\bar a$, the set of atoms fixed by every element of
$\Fix_{\Aut(M)}(\bar a)$ is exactly the set of entries of $\bar a$.
This is a group-theoretic fixed-point condition. In the countable
$\omega$-categorical setting used later in Section~\ref{sec:modeltheory},
it agrees with $\mathrm{dcl}(\bar a)=\bar a$. It should not be
conflated with \emph{degenerate algebraic closure}, the stronger condition
$\mathrm{acl}(\bar a)=\bar a$: definable closure corresponds to
singleton stabilizer orbits, whereas algebraic closure corresponds to
finite stabilizer orbits.
\end{defi}

Automorphism groups of ultrahomogeneous relational structures are closed; in the
countable relational settings used here, the structures considered are
Fra\"iss\'e limits (see, e.g.,
\cite{Macpherson11}); the pure set of any
infinite cardinality is ultrahomogeneous with trivial stabilizer-dcl, since a
bijection between finite subsets extends by any bijection of the
(equinumerous) complements, and a transposition avoiding $\bar a$
moves any prescribed atom outside $\bar a$. We present two automatic
consequences of ultrahomogeneity: the age has the amalgamation property
(amalgamate inside $M$, moving one factor by an automorphism), and $M$
has the \emph{extension property}: every embedding $c\hookrightarrow
d$ in $\mathcal A(M)$ and every embedding $c\hookrightarrow M$ have a
common extension $d\hookrightarrow M$ (embed $d$ anywhere in $M$, then
move its copy of $c$ onto the given one by ultrahomogeneity).

\begin{lem}\label{lem:sitecomp}
Let $M$ be ultrahomogeneous with trivial stabilizer-dcl and put $G=\Aut(M)$.
Let $\mathcal A(M)$ denote the category of finite induced
substructures of $M$ and embeddings. Then
\[
   \mathbb T_G\simeq\mathcal A(M)^{\mathrm{op}},
\]
and the equivalence identifies the atomic topologies. Consequently
the tuple-orbit site depends, up to equivalence, only on the age of
$M$.
\end{lem}

\begin{proof}
Every tuple has the same pointwise stabilizer as the tuple obtained by
deleting repeated entries. Hence every object of $\mathbb T_G$ is
isomorphic to the orbit $O_{\bar c}$ of an injective enumeration of a
finite induced substructure $C\subseteq M$.

Choose an injective enumeration $\bar c$ for each finite induced
substructure $C\subseteq M$. The object $O_{\bar c}$ depends, up to
isomorphism, only on $C$: if $f:C\to D$ is an isomorphism, then by
ultrahomogeneity $f$ extends to some $g\in G$, and conjugation by $g$
identifies the stabilizers of $C$ and $D$, hence identifies the
corresponding transitive $G$-sets.

More generally, let $C,D\subseteq M$ be finite, with chosen enumerations
$\bar c,\bar d$. A morphism $O_{\bar d}\to O_{\bar c}$ is determined by
its value at $\bar d$, say $\lambda\bar c$, and exists exactly when
\[
   \Fix_G(\bar d)\leq\Fix_G(\lambda\bar c).
\]
Trivial stabilizer-dcl implies that every entry of $\lambda\bar c$ occurs
among the entries of $\bar d$; since $\lambda$ is an automorphism, the
coordinate map $c_i\mapsto\lambda c_i$ therefore defines an embedding
$C\hookrightarrow D$. Conversely, if $f:C\hookrightarrow D$ is an
embedding, ultrahomogeneity extends $f$ to some $g\in G$. The assignment
\[
       h\bar d\longmapsto hg\bar c
\]
is then a well-defined $G$-map $O_{\bar d}\to O_{\bar c}$: every element
fixing $\bar d$ fixes the tuple $g\bar c$, whose entries lie in $D$.
The resulting morphism is independent of the chosen extension of $f$,
because any two extensions agree on $C$. These constructions are inverse
and reverse composition, yielding the stated equivalence.

The age has amalgamation: realize two finite extensions in $M$ and use
ultrahomogeneity to identify their common part. Therefore the atomic
topology is defined on both sides, and an equivalence of categories
preserves nonempty sieves and hence identifies the two topologies.
\end{proof}

\begin{rem}[Motivating the trivial stabilizer-dcl]
Without trivial stabilizer-dcl, a morphism out of $O_{\bar d}$ may send
the target tuple into $\mathrm{dcl}(\bar d)$ without placing its
entries in $\bar d$ itself. The ordinary age then no longer determines
the orbit site. The hypothesis is precisely what reduces the
categorical presentation from finite definable closures to finite
substructures.
\end{rem}

\begin{thm}[Age invariance]\label{thm:ageblind}
Let $M$ and $M'$ be ultrahomogeneous relational structures with trivial
stabilizer-dcl and the same age, of arbitrary infinite cardinalities.
Then
\[
\FSS_{\Aut(M)}(M)\;\simeq\;
\mathrm{Sh}\bigl(\mathcal A(M)^{\mathrm{op}},
J_{\mathrm{at}}\bigr)\;\simeq\;\FSS_{\Aut(M')}(M').
\]
More generally, for arbitrary $G\leq\Sym(A)$, the category
$\FSS_G(A)$ is equivalent to the sheaves on the tuple-orbit site of
$\overline G$; in particular it depends only on $\overline G$.
\end{thm}

\begin{proof}
For each structure, Theorem~\ref{thm:site} and Lemma~\ref{lem:sitecomp}
give a presentation by sheaves on the opposite of its finite-substructure
category with the atomic topology. If $M$ and $M'$ have the same age, the
two finite-substructure categories are each equivalent to a skeleton of that
common age (with embeddings); these equivalences preserve nonempty sieves and
therefore the atomic topology. The Comparison Lemma then identifies the two
sheaf toposes. The final assertion for arbitrary $G$ follows from
Theorem~\ref{thm:closure} and Theorem~\ref{thm:site} applied to
$\overline G$.
\end{proof}

\begin{cor}[Pure-atom cardinality invariance]\label{thm:regime}
For every infinite $A$,
\[\FSS_{\Sym(A)}(A)\simeq\Sch\simeq\Nom.\]
where $\Sch$ is the Schanuel topos, the sheaves for the atomic topology on the opposite of the category of finite sets and injections. These are equivalences of Grothendieck topoi. Consequently they preserve finite limits, exponentials, the subobject classifier, and power objects up to
canonical isomorphism. They do not preserve a chosen external
underlying-set functor or external cardinalities, as
Section~\ref{sec:material} makes explicit.
\end{cor}

\begin{proof}
The pure set of cardinality $|A|$ is ultrahomogeneous with trivial
stabilizer-dcl and age the finite sets with injections; apply
Theorem~\ref{thm:ageblind} with $M'$ the countable pure set, for which the identification with $\Sch$ and $\Nom$ is classical~\cite{Johnstone02,Pitts13}.
\end{proof}

\begin{rem}[Attribution]\label{rem:knownness}
We are careful about what is new here. The equivalence
$\mathrm{Cont}(G)\simeq\mathrm{Sh}(\text{open-subgroup site})$ for a
topological group whose open subgroups form a neighbourhood basis is
classical \cite{MacLaneMoerdijk}. More to the point, Caramello's
topological Galois theory~\cite{Caramello16} proves, for a category
$\mathcal C$ of finite structures with a
$\mathcal C$-ultrahomogeneous object~$u$ of its ind-completion, the
equivalence
$\mathrm{Sh}(\mathcal C^{\mathrm{op}},J_{\mathrm{at}})\simeq
\mathrm{Cont}(\Aut(u))$. This is Theorem~\ref{thm:site} together with
Lemma~\ref{lem:sitecomp} when $\mathcal C=\mathcal A(M)$ and $u=M$;
the countable pure-set identification with nominal sets and the
Schanuel topos is also classical \cite{Johnstone02,Pitts13}.
Accordingly, Theorem~\ref{thm:ageblind} is a direct specialization,
not a new topos equivalence: its right-hand side depends on the age,
not on the chosen ultrahomogeneous model. The representation of
orbit-finite sets over homogeneous atoms is developed in \cite{BKL14}.

\noindent
The role of the present section is therefore delimited and of a different character:
\begin{enumerate}[label=\textup{(\roman*)},leftmargin=2.4em]
\item to present, with a direct support-calculus proof, the \emph{closure-invariance statement} (Theorem~\ref{thm:closure}) for an arbitrary $G\leq\Sym(A)$, thereby reducing the two-parameter family $\FSS_G(A)$ to the closed case without adding closedness, oligomorphicity, or $\omega$-categoricity hypotheses;
\item the consequence, made explicit here, that the resulting equivalence is \emph{uniform in $|A|$}, so that the cardinality parameter is categorically inert (Corollary~\ref{thm:regime}); the statement also has genuinely uncountable instances whenever ultrahomogeneous models with the same age exist; for example, the pure set is available at every infinite cardinality, and saturated dense linear orders at cardinalities where such saturated models exist yield further instances;
\item the resulting clarification: raising the cardinality of the atom set, with full symmetry, does not by itself change the category. Any difference must be material (Section~\ref{sec:material}) or carried by the symmetry group (Sections~\ref{sec:restrict}--\ref{sec:join}), and this is what directs the rest of the~paper.
\end{enumerate}
Item (iii) is the reason the section is here: it is what sends the rest of the paper to the symmetry parameter, where Sections~\ref{sec:meet}--\ref{sec:orient} do their work.
\end{rem}

\begin{cor}\label{cor:orbitblind}
Every object of $\FSS_{\Sym(A)}(A)$ with a single orbit of uncountable cardinality corresponds, under the equivalence of Corollary~\ref{thm:regime}, to a single-orbit nominal set of at most countable cardinality. Orbit cardinality is not a categorical invariant of finitely supported sets.
\end{cor}
\begin{proof}
Single-orbit objects are the atoms of the subobject lattice
(Lemma~\ref{lem:subobj}) and are thus preserved by any
equivalence. Every single-orbit nominal set is a quotient of the
orbit of a finite tuple of atoms: if $x$ is supported by a finite tuple
$\bar a$, the equivariant map $\pi\bar a\mapsto\pi x$ is surjective.
Over a countable atom set that tuple orbit, and hence its quotient, is
at most countable. The source orbit is uncountable, so the underlying
cardinalities~differ.
\end{proof}

\begin{rem}[Failure of naive reconstruction]\label{rem:reconfail}
Corollary~\ref{thm:regime} exhibits, for uncountable~$B$, the
non-isomorphic topological groups $\Sym(\omega)$ and $\Sym(B)$ (one
Polish, the other not) with equivalent toposes of continuous discrete
actions. Therefore, the topos $\mathrm{Cont}(G)$ does not determine the
topological group $G$, even among closed permutation groups that are
oligomorphic on their own domains: reconstruction questions must be
normalized to a fixed countable domain, as in Question~\ref{q:reconstruction}. In the Galois-theoretic interpretation, the
two groups are automorphism groups of two different \emph{points} of the same atomic topos, and Theorem~\ref{thm:ageblind} says that the topos itself sees only their common age.
\end{rem}

\section{Material separation}\label{sec:material}

The previous section showed that the cardinality of the atom set is invisible to the category. It does not follow that the choice of cardinality is without content, and this section shows where the content went. An equivalence of categories is free to move objects to non-isomorphic sets; what it cannot preserve, and what distinguishes the frameworks, is exactly the external, set-theoretic data that a material presentation retains.
This dichotomy (that abstract categorical equivalence is blind to external cardinalities while material presentations are not) is what restricts the focus of the remainder of the paper to the symmetry parameter.

\paragraph{Categorical and material notation.}
We reserve $\FSS_G(A)$ for the category of finitely supported $G$-sets and equivariant maps. The corresponding material universe is denoted by $\Vfs^G(A)$: it consists of the hereditarily finitely supported members of the cumulative ZFA hierarchy generated by $A$, with the canonical extension of the $G$-action. Statements about equivalence concern $\FSS_G(A)$; statements about cardinality, choice, and isomorphism of models concern~$\Vfs^G(A)$.

All assertions about $\Vfs^G(A)$ are external assertions in the surrounding
ZFC metatheory. We regard $\Vfs^G(A)$ as the usual hereditarily finitely
supported class model built over the atom sort $A$; when two such universes
are compared as ZFA structures, the atom predicate is part of the signature.
No categorical equivalence below is taken to identify these material models
or their chosen atom sorts.

The equivalence of Corollary~\ref{thm:regime} is an equivalence of abstract categories; it does not commute with the underlying-set functor, and the finitely supported universes over atom sets of different cardinalities are materially distinct. Three invariants make this precise.

\begin{prop}\label{prop:fscofinite}
For $G=\Sym(A)$, a subset $B\subseteq A$ is finitely supported if and
only if it is finite or cofinite.
\end{prop}
\begin{proof}
Finite and cofinite sets are supported by $B$, resp.\ $A\smallsetminus B$. Conversely if $S$ supports~$B$ then $B\smallsetminus S$ is a union of $\Fix_{\Sym(A)}(S)$-orbits on $A\smallsetminus S$ (Lemma~\ref{lem:basic}(iv));
$\Fix_{\Sym(A)}(S)$ is transitive on $A\smallsetminus S$, so $B\smallsetminus
S\in\{\emptyset,A\smallsetminus S\}$.
\end{proof}

\begin{thm}[Material Separation]\label{thm:material}
Let $A$ be infinite and $G=\Sym(A)$.
\begin{enumerate}[label=\textup{(\alph*)}]
\item $|\Pfs(A)|=|A|$; more generally,
$|\Rel^{(n)}_{\fs}(A)|=|A|$ for every $n\geq1$; hence
$|\Pfs(A)|<2^{|A|}$.
\item For every infinite $A$, the set $[A]^2$ of unordered pairs of
atoms admits no finitely supported choice function. In particular,
the canonical epimorphism
\[
 p:E_2:=\{(P,a)\in[A]^2\times A:a\in P\}\longrightarrow[A]^2
\]
has no section in $\FSS_{\Sym(A)}(A)$. The absence of an equivariant
section is categorical and therefore transfers along the equivalences
of Corollary~\ref{thm:regime}; the stronger absence of every finitely
supported choice function is proved directly and uniformly in $|A|$.
\item If $|A|>\aleph_0$, then the material universes
$\Vfs^{\Sym(\omega)}(\omega)$ and $\Vfs^{\Sym(A)}(A)$ are not
isomorphic as ZFA structures with their atom predicate. Already their
atom sorts have different external cardinalities; equivalently, the
definable rank-one object of finitely supported subsets of the atom
sort has external cardinality $\aleph_0$ in the first presentation and
$|A|$ in the second.
\end{enumerate}
\end{thm}

\begin{proof}
(a) For each finite $S\subseteq A$, the group $\Fix(S)$ has only
finitely many orbits on $A^n$, classified by equality patterns among
the coordinates and their incidences with the parameters in $S$.
Every relation supported by $S$ is a union of these finitely many
orbits. Hence only finitely many relations are supported by a fixed
$S$, and
\[
  |\Rel^{(n)}_{\fs}(A)|\leq |[A]^{<\omega}|=|A|.
\]
The reverse inequality follows from the $|A|$ singleton subsets of
$A^n$. The strict inequality is Cantor's theorem.

(b) Suppose $c:[A]^2\to A$ is a choice function supported by the
finite set $S$. Choose distinct $a,b\in A\setminus S$. The
transposition $\tau=(a\ b)$ fixes $S$ and the argument $\{a,b\}$.
Support of the function means
\[
  c(\tau\{a,b\})=\tau c(\{a,b\}),
\]
so $c(\{a,b\})$ would be fixed by $\tau$; however, neither admissible value is. Thus no finitely supported choice function exists. A section of $p$
in the category is an equivariant choice function, so it is excluded as
a special case. Equivalences preserve the existence of sections.

(c) An isomorphism of ZFA structures with the atom predicate restricts
to a bijection between the atom sorts. No such bijection exists when
their external cardinalities differ. Part (a) gives the stated
rank-one witness as well.
\end{proof}

\begin{rem}
Part (b) separates two claims that are easily conflated. The
nonexistence of an equivariant section is an internal categorical
statement and is transported by the cardinality-invariance
equivalence. The
nonexistence of every finitely supported choice function quantifies
also over maps with nonempty support; it is not itself a statement
about morphisms, but the same transposition argument proves it at each
infinite cardinality. Cardinal-sensitive information therefore lies
in the material presentation, whereas the basic choice obstruction is
uniform.
\end{rem}

\begin{rem}[Structural versus material]\label{rem:structuralism}
Corollary~\ref{thm:regime} and Theorem~\ref{thm:material} explain why the
material FSS universe over an uncountable atom reservoir is not made
redundant by categorical equivalence. At the categorical level,
$\FSS_{\Sym(A)}(A)$ is equivalent to the same Schanuel topos for every
infinite~$A$; consequently every statement invariant under equivalence of
these action categories is already determined in the countable
presentation. In particular, categorical questions such as existence of a
section of the choice epimorphism $E_2\to[A]^2$ are transported by the
equivalence.

The material universe $\Vfs^{\Sym(A)}(A)$ retains information that this
abstract action category forgets. The atom predicate has external
cardinality $|A|$, Theorem~\ref{thm:material}(a) gives
$|\Pfs(A)|=|A|$, and Theorem~\ref{thm:material}(b) strengthens the
categorical pair-choice obstruction by ruling out \emph{every} finitely
supported choice function, including those with nonempty support. Thus two
presentations may have equivalent finite-support action topoi while being
non-isomorphic as material ZFA universes. The two viewpoints are therefore
complementary: the categorical level isolates structural invariants of the
action, whereas the material FSS level retains external cardinal and
support-sensitive set-theoretic information.
\end{rem}

\section{Restriction of symmetry}\label{sec:restrict}

We now turn to the symmetry parameter, and begin with the most direct comparison: restrict the group and observe the categorical consequences.
Passing from $G$ to a subgroup~$H$ can only make support easier to achieve, so every $G$-set survives the passage; the question is how much of the categorical structure survives with it. The answer is a useful baseline for everything that follows, and it also shows why a baseline is not enough; the comparison functor is too weak to settle whether the two universes are genuinely different, which is what forces the invariant of Section~\ref{sec:orient}.

Restriction compares data symmetries on the same atom set. It is a
faithful change-of-symmetry functor, but not generally an embedding in the categorical sense: distinct $G$-actions may have the same
restriction, and infinite supported products can enlarge after the
group is reduced.

\begin{thm}[Restriction of symmetry]\label{thm:restrict}
Let $H\leq G\leq\Sym(A)$.
\begin{enumerate}[label=\textup{(\alph*)}]
\item Restriction of actions defines a functor
$\rho^G_H:\FSS_G(A)\to\FSS_H(A)$. Every finite $G$-support remains
an $H$-support; hence, whenever least supports exist,
\[
   \supp_H(x)\subseteq\supp_G(x).
\]
\item The functor is faithful and conservative. It preserves and
reflects finite limits and all small colimits.
\item In general it is neither full nor injective on objects, and it
need not preserve infinite products. On canonical finite powers,
\[
\Sub_{\FSS_H(A)}(A^n)=\Inv_n(H)
\supseteq \Inv_n(G)
=\rho^G_H\bigl(\Sub_{\FSS_G(A)}(A^n)\bigr).
\]
Equality at every finite arity is equivalent to
$\overline H=\overline G$ in the pointwise-convergence topology; in
particular, for closed $H$ and $G$ it is equivalent to $H=G$.
\item If $H\subsetneq G$ are closed subgroups of $\Sym(A)$, the
inclusion in (c) is strict at some finite arity.
\end{enumerate}
\end{thm}

\begin{proof}
Since
$\Fix_H(S)=H\cap\Fix_G(S)\subseteq\Fix_G(S)$, every $G$-support is
an $H$-support. Faithfulness and conservativity follow because the
underlying functions are unchanged and a bijective equivariant map
has an equivariant inverse.

Finite limits in either category are the set-theoretic finite limits
with coordinatewise action; a tuple has the union of the finitely many
component supports. Colimits are the set-theoretic colimits with the
induced action; each element is represented by an element of one
component and inherits its support. Hence restriction preserves these
constructions. It also reflects them: if a cone or cocone becomes a
finite limit or a colimit after restriction, its underlying set cone or
cocone has the corresponding universal property in $\Set$. The unique
set-theoretic mediating map is $G$-equivariant because the structure
maps of the original diagram are $G$-equivariant, so the same universal
property holds before restriction.

The failures are concrete. If $H=\{1\}$ and $G=\Sym(A)$, the natural
$G$-action on $A$ and the trivial $G$-action on the same underlying set
restrict to the same $H$-object, so the functor is not injective on
objects. Every set map is $H$-equivariant, whereas most maps
$A\to A$ are not $G$-equivariant, so the functor is not full. The
product of an infinite family of natural copies of $A$ in
$\FSS_G(A)$ consists only of the finitely supported tuples, whereas in
$\FSS_H(A)=\Set$ it is the full set-theoretic product.

The subobject formula is Lemma~\ref{lem:subobj}. We now prove the
closure criterion without a countability assumption. If closed
subgroups $K,L\leq\Sym(A)$ have the same invariant relations at every
finite arity, then they have the same finite-tuple orbits: the orbit
$L\bar a$ is $L$-invariant, hence $K$-invariant, and so
$K\bar a\subseteq L\bar a$; symmetry gives equality. Given
$g\in L$ and a finite tuple $\bar a$, there is therefore $k\in K$
with $k\bar a=g\bar a$. Thus $g$ lies in the closure of $K$, which
is $K$; hence $L\leq K$, and symmetrically $K=L$.

Applying this to $\overline H$ and $\overline G$ proves (c), since a
group and its closure have the same finitary invariants. Clause (d)
is the contrapositive for proper closed subgroups.
\end{proof}

\begin{rem}[On the cardinality hypothesis]\label{rem:noctble}
Only the \emph{injectivity} half of the Galois correspondence is used above, and that half is cardinality-free, as the proof shows. The descriptive half (which algebras of relations arise as $\Inv_n(G)$, and the reconstruction machinery built on it) does depend on countability in an essential way, and we do not use it; see Question~\ref{q:reconstruction}.
\end{rem}

\begin{cor}[Full symmetry as a faithful baseline]\label{cor:subsumption}
For every infinite $A$ and every $G\leq\Sym(A)$, choose the age-site
equivalence $\Nom\simeq\FSS_{\Sym(A)}(A)$ from
Corollary~\ref{thm:regime} so that the canonical nominal atom object is
carried to $A$, and compose it with restriction. This gives a faithful and
conservative interpretation
\[
   \Nom\longrightarrow\FSS_G(A)
\]
that preserves finite limits and all small colimits.
For a proper closed $G<\Sym(A)$, over any infinite atom set, this interpretation has strictly fewer subobjects on some finite power~$A^n$ than are available in $\FSS_G(A)$; for example, when $G=\Aut(\mathbb Q,<)$, the order relation belongs to $\Sub_{\FSS_G(A)}(A^2)$ but is unavailable under full symmetry.
\end{cor}
\begin{proof}
Use the age-site equivalence induced by the identity equivalence of the category of finite sets and injections; the representable associated with a singleton is carried to the canonical atom object $A$. Compose this equivalence with the restriction functor of Theorem~\ref{thm:restrict}. Its formal properties follow from Theorem~\ref{thm:restrict}(b), and strictness for proper closed $G$ follows from part~(d). The ordered example is immediate because $<$ is $\Aut(\mathbb Q,<)$-invariant but not $\Sym(\mathbb Q)$-invariant.
\end{proof}

\begin{rem}\label{rem:oldphi}
Corollary~\ref{cor:subsumption} replaces an earlier proposed
construction that attempted to act on a nominal set by restricting an
arbitrary permutation to a chosen countable subset. Since the
permutation need not preserve that subset, the construction does not
define an action. Restriction of the acting group is canonical and
requires no auxiliary choice. It also identifies the precise sense in which a structured symmetry
enlarges the family of relational observables: the enlargement is
strict on suitable finite powers, even though the restriction functor
need not be a categorical embedding.
\end{rem}

\section{The meet law and stabilizer amalgamation}\label{sec:meet}

\paragraph{Notation relative to \cite{BKL14}.}
Boja\'nczyk--Klin--Lasota use right actions and write $G_C$ for the
pointwise stabilizer of a finite set $C$. We use left actions and write
$\Fix_G(C)$. Passing from a right action to the corresponding left action
by $g\cdot x:=x\cdot g^{-1}$ leaves pointwise stabilizers unchanged; since a
subgroup is closed under inversion, it also leaves its orbit as a set
unchanged. Thus their subgroup statements translate literally after
$G_C\rightsquigarrow\Fix_G(C)$, and their right orbit $c\cdot H$ is written
$H\cdot c$ below. We make this convention explicit because
\cite[Fact~9.2 and Thm.~9.3]{BKL14} are used in different parts of the
section.

The least-support property and its group-theoretic characterization are already known for arbitrary data symmetries from \cite{BKL14}. The purpose of this section is narrower and structural: to place that theory inside the invariance hierarchy by identifying its relational form as a meet law. This gives a common language in which least supports can be compared with the independent join law of Section~\ref{sec:join} and with the context principles studied later. After presenting the correspondence and the sharper BKL local test, we contrast equality symmetry with natural symmetries for which the meet law fails.

Both lattice laws compare, for finite $S,T\subseteq A$, the algebras
$B^{(n)}_{G,S}$, $B^{(n)}_{G,T}$ with those at $S\cap T$ and
$S\cup T$. Two inclusions are automatic from antitonicity of
$H\mapsto\Inv_n(H)$:
\begin{equation}\label{eq:trivial}
B^{(n)}_{G,S\cap T}\ \subseteq\ B^{(n)}_{G,S}\cap B^{(n)}_{G,T},
\qquad
B^{(n)}_{G,S}\vee B^{(n)}_{G,T}\ \subseteq\ B^{(n)}_{G,S\cup T}.
\end{equation}
The laws assert the reverse inclusions.

\begin{defi}\label{def:sap}
$G\leq\Sym(A)$ has the \emph{stabilizer amalgamation property}
(\textup{SAP}) if for all finite $S,T\subseteq A$, \ $\bigl\langle\,\Fix_G(S)\cup\Fix_G(T)\,\bigr\rangle\;=\;\Fix_G(S\cap T)$.
It has the \emph{Galois} \textup{SAP} if the formally weaker identity
$\Inv_n\bigl(\langle\Fix_G(S)\cup\Fix_G(T)\rangle\bigr)
=\Inv_n\bigl(\Fix_G(S\cap T)\bigr)$ holds for all~$n\geq1$.
\end{defi}

\begin{prop}\label{prop:sapmeet}
The meet law at all arities is equivalent to the Galois \textup{SAP},
and is implied by the \textup{SAP}.
\end{prop}
\begin{proof}
For any subgroups $H,K\leq G$ one has
$\Inv_n(H)\cap\Inv_n(K)=\Inv_n(\langle H\cup K\rangle)$ directly from
the definitions. With $H=\Fix_G(S)$, $K=\Fix_G(T)$ the meet law reads $\Inv_n(\langle H\cup K\rangle)=\Inv_n(\Fix_G(S\cap T))$, which is the Galois SAP; the SAP gives it by applying $\Inv_n$.
\end{proof}

The meet law thus lives at the level of relation algebras. The next
theorem identifies it with a
property of the entire finitely supported universe: the classical
\emph{least-support property}, the foundation stone of the nominal
calculus, which for general symmetry parameters can~fail.

\begin{thm}[Relational form of the least-support criterion]\label{thm:meetchar}
For every $G\leq\Sym(A)$ the following are equivalent.
\begin{enumerate}[label=\textup{(\arabic*)}]
\item The meet law holds at every arity: for all finite
$S,T\subseteq A$ and all $n\geq1$,
$B^{(n)}_{G,S}\cap B^{(n)}_{G,T}=B^{(n)}_{G,S\cap T}$.
\item The \emph{least-support property} holds: for every $G$-set $X$
and every $x\in X$, if the finite sets~$S$ and $T$ both support $x$
then so does $S\cap T$; consequently every finitely supported element
of every $G$-set has a least support.
\item The Galois \textup{SAP} holds.
\item The \textup{SAP} holds: \
$\bigl\langle\Fix_G(S)\cup\Fix_G(T)\bigr\rangle =\Fix_G(S\cap T)
\ \text{for all finite }S,T\subseteq A$.
\end{enumerate}
The equivalence of \textup{(2)} and \textup{(4)} is the subgroup criterion of \cite[Fact~9.2]{BKL14}; the equivalences with \textup{(1)} and \textup{(3)} are its relation-algebra formulation.
\end{thm}

\begin{proof}
(1)$\Leftrightarrow$(3) is Proposition~\ref{prop:sapmeet}, and
(4)$\Rightarrow$(3) is immediate.
(2)$\Rightarrow$(1): a relation $R\in B^{(n)}_{G,S}\cap B^{(n)}_{G,T}$ is an element of the $G$-set $\mathcal P(A^n)$ (with the image action) supported by $S$ and by $T$; by (2) it is supported by $S\cap T$, i.e.\ $R\in B^{(n)}_{G,S\cap T}$.

(1)$\Rightarrow$(2): let $x\in X$ be supported by $S$ and by $T$. If $S=\emptyset$, then $S\cap T=\emptyset$ already supports $x$. Otherwise put $V:=\mathrm{Stab}_G(x)$, so that
$\Fix_G(S)\cup\Fix_G(T)\subseteq V$. Enumerate~$S$ as a tuple
$\bar a$ of length $k:=|S|$ and consider the relation
$R\;:=\;V\cdot\bar a\;\subseteq\;A^{k}$.
For $\pi\in\Fix_G(S)\cup\Fix_G(T)\subseteq V$ we have
$\pi\cdot R=\pi V\bar a=V\bar a=R$; hence
$R\in B^{(k)}_{G,S}\cap B^{(k)}_{G,T}$. By the meet law at arity $k$,
$R\in B^{(k)}_{G,S\cap T}$, i.e.\ every $\pi\in\Fix_G(S\cap T)$
satisfies $\pi\cdot R=R$. For such $\pi$, in particular
$\pi\bar a\in R$, so $\pi\bar a=v\bar a$ for some $v\in V$; then
$v^{-1}\pi$ fixes the entries of $\bar a$ pointwise, i.e.\
$v^{-1}\pi\in\Fix_G(S)\subseteq V$, whence $\pi\in V$. Thus
$\Fix_G(S\cap T)\subseteq\mathrm{Stab}_G(x)$: the set $S\cap T$
supports $x$.

To obtain a least support, fix one finite support $S_0$ of $x$ and
consider
\[
   \mathcal S_0=\{U\subseteq S_0:U\text{ supports }x\}.
\]
This is a nonempty finite family, and repeated application of the
intersection property shows that
$L:=\bigcap\mathcal S_0$ supports $x$. If $T$ is any finite support,
then $S_0\cap T$ is a member of $\mathcal S_0$, so
$L\subseteq S_0\cap T\subseteq T$. Hence $L$ is the least support.

For reference, we also present the direct implication (2)$\Rightarrow$(4), although the equivalence of these two conditions is exactly \cite[Fact~9.2]{BKL14}. Fix finite $S,T$ and put
$H:=\bigl\langle\Fix_G(S)\cup\Fix_G(T)\bigr\rangle$.
In the left-coset $G$-set $G/H$, the base point $H$ is supported both
by $S$ and by $T$. Hence it is supported by $S\cap T$, so
$\Fix_G(S\cap T)\leq\mathrm{Stab}_G(H)=H$. The reverse inclusion is
automatic because both generators fix $S\cap T$ pointwise. Thus
$H=\Fix_G(S\cap T)$.
\end{proof}

Theorem~\ref{thm:meetchar} should therefore be read as a change of language, not as a new least-support test: it identifies the meet law of the invariance hierarchy with the already known subgroup criterion and presents the equivalence in the notation used throughout the rest of the paper.

\begin{cor}[BKL subgroup criterion, in the present notation]\label{cor:criterion}
Let $A$ be an infinite set and $G\leq\Sym(A)$ any permutation group.
Then every finitely supported element of every $G$-set has a least
support if and only if
\begin{equation}\label{eq:sapcrit}
\Fix_G(S\cap T)\;=\;\bigl\langle\,\Fix_G(S)\cup\Fix_G(T)\,\bigr\rangle
\qquad\text{for all finite }S,T\subseteq A.
\end{equation}
Equivalently, the meet law and the Galois \textup{SAP} hold. As in the original criterion of \cite{BKL14}, this statement is about arbitrary data symmetries; no additional closedness or oligomorphicity assumption is being introduced here.
\end{cor}

\begin{proof}
This is Theorem~\ref{thm:meetchar}.
\end{proof}

\begin{cor}[Unary elementary-diamond reduction]\label{cor:arityone}
For $G\leq\Sym(A)$ the following are equivalent.
\begin{enumerate}[label=\textup{(\arabic*)}]
\item The meet law holds at every arity.
\item For every finite $E\subseteq A$ and distinct $c,d\in A\setminus E$,
\[
B^{(1)}_{G,E\cup\{c\}}\cap B^{(1)}_{G,E\cup\{d\}}=B^{(1)}_{G,E}.
\]
\item For every such $E,c,d$, with
\[
H_{E,c,d}:=\bigl\langle\Fix_G(E\cup\{c\})\cup\Fix_G(E\cup\{d\})\bigr\rangle,
\]
the BKL local orbit condition holds in left-action notation:
\[
\Fix_G(E)\cdot c\ \subseteq\ H_{E,c,d}\cdot c.
\]
(The reverse containment is automatic because $H_{E,c,d}\leq\Fix_G(E)$.)
\end{enumerate}
\end{cor}
\begin{proof}
(1)$\Rightarrow$(2) is immediate. For (2)$\Rightarrow$(3), Proposition~\ref{prop:sapmeet} gives
\[
B^{(1)}_{G,E\cup\{c\}}\cap B^{(1)}_{G,E\cup\{d\}}=\Inv_1(H_{E,c,d}).
\]
Thus (2) says $\Inv_1(H_{E,c,d})=\Inv_1(\Fix_G(E))$. Unary invariant algebras determine the orbit partition of $A$: their minimal nonempty members are exactly the corresponding orbits. Hence the $H_{E,c,d}$-orbit and the $\Fix_G(E)$-orbit of $c$ coincide, which is (3). Finally, (3)$\Rightarrow$ least supports is exactly \cite[Thm.~9.3]{BKL14}, translated from the right-action notation used there to the present left-action convention, and least supports imply (1) by Theorem~\ref{thm:meetchar}.
\end{proof}

\begin{cor}[Hierarchy form of BKL well-behavedness]\label{cor:wellbehaved}
Let $(A,G)$ be a Fra\"iss\'e symmetry in the sense of \cite[Sec.~10]{BKL14}. Then it is well behaved in the sense of \cite[Def.~10.5]{BKL14} if and only if its invariance hierarchy satisfies both:
\begin{enumerate}[label=\textup{(\alph*)}]
\item the meet law (equivalently, the unary elementary-diamond law of Corollary~\ref{cor:arityone});
\item injectivity of $\ell_n$ for some, equivalently every, $n\geq1$.
\end{enumerate}
\end{cor}
\begin{proof}
By \cite[Def.~10.5]{BKL14}, a Fra\"iss\'e symmetry is well behaved exactly when it admits least supports and is fungible. Theorem~\ref{thm:meetchar} identifies the first condition with the meet law, and Lemma~\ref{lem:fungible} identifies the second with injectivity of the level maps.
\end{proof}

\begin{rem}[Relation to the BKL local criterion]\label{rem:bkl-local}
Corollary~\ref{cor:arityone} is not an independent new least-support criterion: its implication from the local orbit condition is precisely \cite[Thm.~9.3]{BKL14}. Its point is relational. It shows that preservation of meets by the entire family of invariant-relation algebras, at all arities and all finite pairs of contexts, is already determined by the unary levels on elementary diamonds $E\subset E\cup\{c\},E\cup\{d\}$. Together with Lemma~\ref{lem:fungible}, this turns the two defining conditions of BKL well-behaved Fra\"iss\'e symmetries into two elementary structural properties of one lattice-valued invariant.
\end{rem}

\begin{lem}[Equivariance of least supports]\label{lem:leastequi}
If every finitely supported element of every $G$-set has a least
support, then $\supp(\pi\cdot x)=\pi(\supp(x))$ for all $\pi\in G$.
\end{lem}

\begin{proof}
$\pi(\supp x)$ supports $\pi x$ by Lemma~\ref{lem:basic}(v); and if
$T$ supports $\pi x$ then $\pi^{-1}(T)$ supports~$x$, so
$\pi^{-1}(T)\supseteq\supp x$, i.e.\ $T\supseteq\pi(\supp x)$. Hence
$\pi(\supp x)$ is the least support of~$\pi x$.
\end{proof}

\begin{rem}[Position in the literature]\label{rem:leastlit}
The least-support property is a foundational hypothesis of nominal technique and is treated for arbitrary data symmetries in \cite{BKL14}. Definition~4.11 there is the same intersection-closure property as Theorem~\ref{thm:meetchar}(2); Fact~9.2 is equivalent to \eqref{eq:sapcrit}; and Theorem~9.3 gives the sharper local single-orbit test recorded in Remark~\ref{rem:bkl-local}. Accordingly, neither the generality in $G$ nor the subgroup criterion is claimed as new here. The difference in quantification is only notational: BKL formulate the property for nominal $G$-sets, whereas Theorem~\ref{thm:meetchar}(2) mentions finitely supported elements of arbitrary $G$-sets; every such element lies in its nominal orbit $G\cdot x$, so the two formulations are equivalent. For $\Sym(A)$, support intersection also holds at every infinite cardinality in the material FSS treatment \cite[Chapter~2]{AC20}.

The contribution of Theorem~\ref{thm:meetchar} is the relation-algebra formulation: it identifies the already known least-support condition with preservation of meets by the finite-context invariant algebras. This formulation is useful because the same hierarchy supports the independent join law of Section~\ref{sec:join}, the freshness analysis of Section~\ref{subsec:someany}, and the model-theoretic interpretation of Proposition~\ref{prop:weakEI}. Related permutation-group formulations of weak elimination of imaginaries are used in \cite{Paolini24}.
\end{rem}

The BKL examples also calibrate the two hierarchy properties just isolated.

\begin{prop}[BKL pathologies in the hierarchy]\label{prop:bklpathologies}
The standard examples of \cite[Examples~9.1, 9.9 and 9.10]{BKL14} have the following hierarchy profiles.
\begin{enumerate}[label=\textup{(\alph*)}]
\item For $A=\mathbb N\times\mathbb N$ and $G=\{\pi^2:\pi\in\Sym(\mathbb N)\}$, where $\pi^2(i,j)=(\pi(i),\pi(j))$, the meet law fails already at arity~$1$. This is a closed oligomorphic example with nontrivial stabilizer-definable closure.
\item For a countably infinite $A$ with distinguished point $d$ and $G=\Sym(A)_d$, the meet law holds but the level maps are not injective.
\item For $A=\{0,1\}\times\mathbb N$ with the two-layer symmetry of \cite[Example~9.10]{BKL14}, the level maps are injective but the meet law fails already at arity~$1$.
\end{enumerate}
Consequently meet preservation and level injectivity (equivalently, least supports and fungibility) are independent.
\end{prop}
\begin{proof}
For (a), put $p=(0,1)$ and $q=(1,0)$. The unary relation $\{p\}$ is supported both by $\{p\}$ and by $\{q\}$: fixing $q$ forces the underlying permutation of $\mathbb N$ to fix both $0$ and $1$, hence also $p$. It is not $\emptyset$-supported. Thus
\[
\{p\}\in B^{(1)}_{G,\{p\}}\cap B^{(1)}_{G,\{q\}}\setminus B^{(1)}_{G,\emptyset}.
\]
The group is oligomorphic because orbits of finite tuples of pairs are determined by equality types among finitely many underlying coordinates. It is closed: if a permutation of $A$ is a pointwise limit of diagonal permutations, its values on $(i,i)$ determine a single permutation~$\pi$ of~$\mathbb N$, and approximation simultaneously on $(i,i),(j,j),(i,j)$ forces the value of $(i,j)$ to be $(\pi(i),\pi(j))$. Finally, $\Fix_G(p)$ also fixes $q\neq p$, so the stabilizer-definable closure of $p$ is nontrivial.

For (b), BKL show that the distinguished-point symmetry admits least supports, hence the meet law holds by Theorem~\ref{thm:meetchar}. But for every $e\neq d$,
$\Fix_G(\{d,e\})=\Fix_G(\{e\})$, so Lemma~\ref{lem:fungible} gives failure of level injectivity.

For (c), BKL show that the two-layer symmetry is fungible, hence the level maps are injective by Lemma~\ref{lem:fungible}. Let $L=\{0\}\times\mathbb N$. Fixing any atom $(0,n)$ rules out the global layer swap, so $\{(0,n)\}$ supports $L$; but $L$ is not $\emptyset$-supported because an element of $G$ may swap the two layers. Distinct singletons therefore witness failure of the unary meet law.
\end{proof}

The following example is deliberately not presented as evidence of greater scope: it is the integer symmetry already used throughout \cite[Examples~2.3, 4.7 and 4.12]{BKL14}. We revisit it only to display, in the simplest possible form, what failure of the meet law looks like inside the invariance hierarchy.

\begin{exa}[The integer symmetry in relational form]\label{exa:offset}
Let $A=\mathbb Z$ and let $G=\langle t\rangle$ be the infinite cyclic
group generated by the shift $t(n)=n+1$, acting regularly. This is the
symmetry appropriate to values on which a program may compute
\emph{differences} but not absolute positions: relative timestamps,
offsets into a stream, cyclic sequence numbers. Its invariant binary
relations are exactly the unions of difference classes
$\{(x,y):y-x=d\}$ for $d \in \mathbb{Z}$, so $G$ has infinitely many orbits on $A^{2}$ and is not locally oligomorphic and is not the automorphism group of an $\omega$-categorical structure. The model-theoretic route through
Proposition~\ref{prop:weakEI} is therefore unavailable, whereas the
group criterion and the lattice characterization apply directly.

The failure of least supports for this integer symmetry is exactly the standard BKL counterexample \cite[Examples~4.7 and 4.12]{BKL14}. What is added here is only its relational visualization in the invariance hierarchy. The action is
regular, so $\Fix_G(S)=\{1\}$ for every nonempty finite $S$, whence
\[
   \bigl\langle\Fix_G(\{0\})\cup\Fix_G(\{1\})\bigr\rangle=\{1\}
   \neq G=\Fix_G(\emptyset)=\Fix_G(\{0\}\cap\{1\}),
\]
so the \textup{SAP} fails and least supports do not exist. The failure
is visible directly in the hierarchy: $B^{(n)}_{G,S}$ is the algebra of
\emph{all} $n$-ary relations for every nonempty $S$, while
$B^{(n)}_{G,\emptyset}$ contains only the shift-invariant ones, so the
meet law fails at $S=\{0\}$, $T=\{1\}$ already at arity~$1$ (for
example, the singleton $\{0\}$ belongs to both nonempty-context levels
but not to the empty-context level). Concretely,
any state whose orbit is nontrivial is supported by $\{0\}$ and by
$\{1\}$ and by no smaller context, so a configuration of such a machine
has no canonical register contents: every single value is an equally
good origin, and none is forced.
\end{exa}

\begin{rem}\label{rem:offsetreading}
The difference symmetry is natural, but its support discipline is degenerate in a precise sense: $\{0\}$ is a finite base for $G$, so every element of every $G$-set is finitely supported. We use the example only as a transparent illustration of the relational meet failure, not as a new counterexample or as evidence for a broader least-support criterion. It also shows that the obstruction can be visible without invoking imaginaries in the model-theoretic sense: here the symmetry acts freely, so fixing any one value fixes everything. In contrast, full equality symmetry has least supports at every infinite cardinality by the classical support-intersection lemma.
\end{rem}

For full equality symmetry, finite supports are closed under
intersection at every infinite atom cardinality \cite{AC20}.
Theorem~\ref{thm:meetchar} therefore yields both the meet law and the
\textup{SAP} for $\Sym(A)$. This classical case is the nominal baseline; the point here is not to re-characterize least supports for arbitrary data symmetries, which is already done in \cite{BKL14}, but to expose that property and its failures inside the invariant-relation hierarchy.

\begin{thm}[Meet-law failure for the generic equivalence relation]
\label{thm:equivfail}
Let $(A,E)$ be an equivalence structure with infinitely many classes,
all of the same infinite cardinality, and let $G=\Aut(A,E)$. Then the
meet law fails at arity $1$: there are $a\neq b$ with
\[
B^{(1)}_{G,\{a\}}\cap B^{(1)}_{G,\{b\}}\;\supsetneq\;
B^{(1)}_{G,\emptyset}.
\]
Consequently $G$ has neither the \textup{SAP} nor the Galois
\textup{SAP}.
\end{thm}

\begin{proof}
Choose $a\neq b$ in the same class $C=[a]_E=[b]_E$. The class $C$ is
supported by $\{a\}$: any automorphism fixing $a$ maps $[a]_E$ to
itself. Likewise $C$ is supported by $\{b\}$. But $C$ is not
$\emptyset$-supported: $G$ acts transitively on the set of classes,
since for classes $C_1\neq C_2$ any bijection $C_1\to C_2$, together
with its inverse and the identity elsewhere, is an automorphism
exchanging them (here the equinumerosity of the classes is used)
so some $\pi\in G=\Fix_G(\emptyset)$ moves~$C$. Thus $C\in
B^{(1)}_{G,\{a\}}\cap B^{(1)}_{G,\{b\}}\smallsetminus B^{(1)}_{G,\emptyset}$. The group-theoretic failure is visible directly: both $\Fix_G(a)$ and $\Fix_G(b)$ preserve $C$ setwise, hence so does the subgroup they generate, which is therefore proper in $G$.
\end{proof}

\begin{rem}[Requirement for Class Equinumerosity]\label{rem:whyequinum}
The hypothesis is not decorative. If the classes have pairwise
distinct infinite cardinalities, then every automorphism of $(A,E)$
fixes every class setwise, each class is $\emptyset$-supported, and
the witness above disappears; the structure is then also not
ultrahomogeneous. The standing hypothesis (infinitely many classes, all of one and the same infinite cardinality) is more than enough to guarantee the
ultrahomogeneity required in Definition~\ref{def:generic}: a finite
partial isomorphism can be extended class by class, and then on the
remaining classes by bijections of equal-sized complements. For
countable $A$ this gives the Fra\"iss\'e limit of the class of finite
equivalence structures, the ``generic equivalence relation'' of the
model-theoretic literature. We use this hypothesis throughout;
we do not claim that it is the only ultrahomogeneous equivalence~structure.
\end{rem}

\section{Freshness sorts and orbit-finiteness}
\label{sec:fresh}

Least supports are the first of the standing hypotheses that nominal
technique rests on, and Section~\ref{sec:meet} located it exactly in the
symmetry parameter. Two further principles are used just as constantly,
and this section locates them in the same way. The first is the
\emph{freshness quantifier}, the \textsf{some/any} principle of nominal
logic \cite{Pitts13}, which licenses the reasoning step
``choose a fresh name; the result does not depend on which.'' The
second is \emph{orbit-finiteness}, the finiteness notion on which the
theory of automata over atoms is built \cite{BKL14}. Both turn out to
be governed by the same two-line calculus of orbits, and the answers are
sharper than the least-support answer: the freshness quantifier in its
unrestricted form characterizes the equality symmetry outright, while
orbit-finiteness decomposes into finiteness of every context level and a
bound on support size.

\subsection{The some/any principle and its graded form}
\label{subsec:someany}

Throughout this subsection $S$ ranges over finite subsets of $A$. For a
$G$-set $X$ write
\[
   X^{S}\;=\;\{\,x\in X:\pi\cdot x=x\text{ for all }\pi\in\Fix_G(S)\,\}
\]
for the set of elements supported by $S$. Thus $B^{(n)}_{G,S}$ is
$\bigl(\mathcal P(A^n)\bigr)^{S}$, and Lemma~\ref{lem:basic}(iv) says
that a member of $X^{S}$ is, when $X$ is a set of subsets, a union of
$\Fix_G(S)$-orbits.

\begin{defi}[Freshness sorts]\label{def:sorts}
The \emph{freshness sorts over $S$} are the orbits of $\Fix_G(S)$ on
$A\smallsetminus S$. Write $\mathrm{sorts}_G(S)$ for their cardinality and
$ s_G(k)\;=\;\sup\{\,\mathrm{sorts}_G(S):|S|=k\,\}$.
We say $G$ has the \emph{some/any property} if $s_G(k)=1$ for every
$k$, that is, if $\Fix_G(S)$ acts transitively on $A\smallsetminus S$
for every finite $S$.
\end{defi}

The name is justified by the following proposition, which is the exact
form in which freshness reasoning is used: a supported predicate cannot
distinguish two atoms lying in the same sort, so on each sort
``for some fresh atom'' and ``for every fresh atom'' agree.

\begin{prop}[Graded \textsf{some/any}]\label{prop:someany}
Let $G\leq\Sym(A)$, let $S$ be finite and let $\varphi\subseteq A$ be
supported by $S$. Then $\varphi\smallsetminus S$ is a union of
freshness sorts over $S$. Consequently, for every sort $\sigma$ over
$S$, \quad $ (\exists a\in\sigma)\ a\in\varphi
   \qquad\Longleftrightarrow\qquad
   (\forall a\in\sigma)\ a\in\varphi$.

More generally, let $X\in\FSS_G(A)$, let $\varphi\subseteq A\times X$ be
supported by $S$ and let $x\in X$ be supported by $T$. Then the section
$\varphi_x=\{a\in A:(a,x)\in\varphi\}$ is supported by $S\cup T$, so the
displayed equivalence holds for every sort over $S\cup T$.
\end{prop}

\begin{proof}
The first claim is Lemma~\ref{lem:basic}(iv) applied to
$\varphi\in\bigl(\mathcal P(A)\bigr)^{S}$, together with the observation
that each $a\in S$ is its own $\Fix_G(S)$-orbit. The displayed
equivalence follows since $\varphi$ contains a sort or is disjoint from
it. For the relativized form, let $g\in\Fix_G(S\cup T)$. Since $S$
supports $\varphi$ and $T$ supports $x$, we have $g\varphi=\varphi$ and
$gx=x$. Therefore $g\varphi_x=(g\varphi)_{gx}=\varphi_x$. Thus
$S\cup T$ supports $\varphi_x$, and the first part applies to it.
\end{proof}

Proposition~\ref{prop:someany} requires no hypothesis whatever on $G$:
freshness reasoning is always available \emph{sort by sort}. What is
special about the nominal setting is that there is only one sort, and
this turns out to characterize it completely.

\begin{thm}[The freshness quantifier characterizes equality atoms]
\label{thm:someanychar}
Let $G\leq\Sym(A)$ be closed. Then $G$ has the some/any property if
and only if $G=\Sym(A)$.
\end{thm}

\begin{proof}
If $G=\Sym(A)$ then $\Fix_G(S)$ is the full symmetric group of
$A\smallsetminus S$, which is transitive on that set.

Conversely, suppose $\Fix_G(S)$ is transitive on $A\smallsetminus S$ for
every finite $S$. We show by induction on $k$ that $G$ is
$k$-transitive. For $k=1$ this is the case $S=\emptyset$. Assume $G$
is $k$-transitive and let $(a_1,\dots,a_{k+1})$ and
$(b_1,\dots,b_{k+1})$ be injective tuples. By the inductive hypothesis
some $\pi\in G$ carries $(a_1,\dots,a_k)$ to $(b_1,\dots,b_k)$; put
$S=\{b_1,\dots,b_k\}$. Then $\pi(a_{k+1})$ and $b_{k+1}$ both lie in
$A\smallsetminus S$, so some $\rho\in\Fix_G(S)$ carries the first to the
second, and~$\rho\pi$ carries the longer tuple as required. Hence $G$
is $k$-transitive for every~$k$, that is, $G$ is highly transitive and
therefore dense in $\Sym(A)$ for the topology of pointwise convergence.
A closed dense subgroup is the whole group.
\end{proof}

Closedness is essential in Theorem~\ref{thm:someanychar}. The some/any
property forces $G$ to be highly transitive and hence dense in $\Sym(A)$,
but density alone does not force equality. For example, the finitary
symmetric group is a proper dense subgroup with the same finite-context orbit
structure. Equivalently, for arbitrary $G$ the property determines the
pointwise closure $\overline G=\Sym(A)$.

\begin{rem}[Interpreting the theorem]\label{rem:someanyreading}
Theorem~\ref{thm:someanychar} should be read carefully, because it is
not a deficiency of the other symmetries. It says that the \emph{unrelativized}
freshness quantifier is unavailable as soon as programs may observe any
structure at all on the data, and that what replaces it is the palette
of sorts of Definition~\ref{def:sorts}. The palette is finite under
local oligomorphicity, but Proposition~\ref{prop:someany} itself does
not require finiteness. This is precisely what
happens in practice: over ordered atoms one does not choose a fresh
timestamp, one chooses a fresh timestamp \emph{in a given interval}; over
graph atoms one chooses a fresh vertex \emph{with a prescribed adjacency
pattern to the registers}. The sort count $s_G(k)$ measures how many
distinct kinds of fresh value an implementation must be able to
generate, and Proposition~\ref{prop:someany} certifies that within a
sort the classical reasoning step is sound unchanged.
\end{rem}

For locally oligomorphic symmetries the palette is finite and easy to
compute, and the computation separates the standard data symmetries at a
glance.

\begin{prop}[Sort counts]\label{prop:sortcounts}
Let $|S|=k$. Then
\begin{enumerate}[label=\textup{(\roman*)}]
\item $s_{\Sym(A)}(k)=1$ for every $k$ and every infinite $A$;
\item for the dense order, $s_{\Aut(\mathbb Q,<)}(k)=k+1$, the sorts
being the intervals cut by $S$;
\item for the random graph $\mathbb G$,
$s_{\Aut(\mathbb G)}(k)=2^{k}$, the sorts being the adjacency patterns
to $S$;
\item for the generic equivalence relation, a context $S$ meeting $m$
equivalence classes has $m+1$ sorts, and hence
$s_{\Aut(A,E)}(k)=k+1$;
\item for the generic $3$-uniform hypergraph,
$s(k)=2^{\binom{k}{2}}$;
\item for the betweenness reduct of the dense order,
the sort count is $1$ for $k=0,1$ and $k+1$ for $k\geq2$.
\end{enumerate}
The counts in \textup{(ii)--(iv)} are $>1$ for $k\geq1$, while those
in \textup{(v)} and \textup{(vi)} are $>1$ for $k\geq2$. The
small-context exceptions are consistent with
Theorem~\ref{thm:someanychar}, whose hypothesis concerns every finite
context.
\end{prop}

\begin{proof}
In each case the structure is ultrahomogeneous, so two atoms outside $S$
lie in the same $\Fix_G(S)$-orbit exactly when they realize the same
quantifier-free type over $S$; the counts are the numbers of such types
realized in the respective homogeneous structures. For (i) there is one
type, that of an atom distinct from each member of $S$. For (ii) the
type of $a\notin S$ over $S$ indicates its position among the $k$ named
points, giving $k+1$ possibilities, all realized by density. For (iii)
the type indicates the adjacency pattern to $S$, and all $2^k$ patterns are realized by the extension property of $\mathbb G$. For (iv) the type records which class of a member of $S$ contains $a$, if any, giving
$m+1$ possibilities when $S$ meets $m$ classes; the supremum is $k+1$
because $S$ may meet $k$ distinct classes. For (v) the type indicates
which of the $\binom k2$ pairs from $S$ form a hyperedge with $a$, and
all patterns are realized by the extension property. For (vi), when
$k\geq2$, write the members of $S$ in increasing order. Betweenness
with pairs from $S$ distinguishes the $k+1$ intervals they cut, while
the pointwise stabilizer is transitive on each interval. For $k=0,1$
the complement of $S$ is a single orbit.
\end{proof}

\subsection{Orbit-finiteness is bounded uniform finiteness}
\label{subsec:orbitfinite}

The finiteness notion of the theory of automata over atoms is
orbit-finiteness: an object is admissible as a state space when it is a
finite union of $G$-orbits \cite{BKL14}. The context filtration
$S\mapsto X^S$ supplies a second, weaker notion, obtained by asking each
fixed-point level to be finite rather than the whole object to be
finitely generated.

\begin{defi}\label{def:unifinite}
Following \cite{AC22}, call $X\in\FSS_G(A)$ \emph{uniform finite} if
$X^{S}$ is finite for every finite $S\subseteq A$. Equivalently (in
the ambient classical metatheory), $X$ has no infinite subset all of
whose elements share one finite support.
\end{defi}

Uniform finiteness says that only finitely many elements are visible
from any one register assignment; orbit-finiteness says in addition that
the register assignments needed are of bounded size. That is exactly
the difference.

\begin{thm}[Orbit-finiteness $=$ uniform finiteness $+$ bounded support]
\label{thm:orbitfinite}
Assume throughout that $G\leq\Sym(A)$ has the least-support property,
and let $X\in\FSS_G(A)$.
\begin{enumerate}[label=\textup{(\arabic*)}]
\item If $X$ is orbit-finite then $X$ is uniform finite and
$\sup\{|\supp(x)|:x\in X\}<\infty$.
\item If, in addition, $G$ is locally oligomorphic, then uniform
finiteness of $X$ together with a uniform bound on its support sizes
implies that $X$ is orbit-finite.
\end{enumerate}
Neither condition on $X$ in \textup{(2)} can be omitted. An infinite
set with the trivial $G$-action has support bound zero but is not
uniform finite and has infinitely many orbits. In the other direction, for $G=\Sym(A)$ the disjoint union $\coprod_{n\in\mathbb N}A^{(n)}$ of the sets of injective $n$-tuples is uniform finite (since any fixed context $S$ can only support elements in summands $A^{(n)}$, where $n\leq |S|$), has unbounded supports, and has infinitely many orbits.
\end{thm}

The two directions use different ideas. A finite union of orbits has a
uniform support bound because least supports move equivariantly along each
orbit; a fixed context then sees only finitely many points of each orbit.
Conversely, a support bound reduces all possible least supports to finitely
many $G$-orbits of finite subsets, and finite context levels supply only
finitely many states above representatives of those support orbits.

\begin{proof}
The case $X=\emptyset$ is immediate, so assume $X\neq\emptyset$.
For (1), write $X=G\cdot x_1\cup\dots\cup G\cdot x_m$ and
$S_i=\supp(x_i)$. By Lemma~\ref{lem:leastequi} every element of
$G\cdot x_i$ has support $\pi(S_i)$ for some $\pi$, hence of size
$|S_i|$; so support sizes are bounded by $\max_i|S_i|$.

For uniform finiteness fix a finite $S$. Every $x\in X^S$ has
$\supp(x)\subseteq S$, so we count states in two finite stages: first by
the ambient orbit $G\cdot x_i$, and then by their least support
$T\subseteq S$.

Fix $i$ and $T\subseteq S$. A state in $G\cdot x_i$ with least support
$T$ has the form $\pi x_i$ with $\pi(S_i)=T$. If no such $\pi$ exists,
there is nothing to count. Otherwise choose one $\pi_0$ with
$\pi_0(S_i)=T$. Then
\[
   \{\pi\in G:\pi(S_i)=T\}=\pi_0G_{\{S_i\}},
\]
where $G_{\{S_i\}}$ is the setwise stabilizer of $S_i$. Two elements
$\pi x_i$ and $\rho x_i$ in this family are equal exactly when
$\rho^{-1}\pi\in U_i:=\mathrm{Stab}_G(x_i)$. Because $S_i$ supports
$x_i$, we have $\Fix_G(S_i)\leq U_i$. Conversely, equivariance of least
supports (Lemma~\ref{lem:leastequi}) implies that every element of $U_i$
preserves $S_i$ setwise, so
\[
   \Fix_G(S_i)\leq U_i\leq G_{\{S_i\}}.
\]
Hence the number of distinct states with least support $T$ in the orbit
$G\cdot x_i$ is
\[
   [G_{\{S_i\}}:U_i]
   \leq [G_{\{S_i\}}:\Fix_G(S_i)].
\]
Restriction to $S_i$ embeds
$G_{\{S_i\}}/\Fix_G(S_i)$ into the finite group $\Sym(S_i)$, so this
index is at most $|S_i|!$. There are at most $2^{|S|}$ possible subsets
$T\subseteq S$. Summing first over these supports and then over the
finitely many ambient orbits gives
\[
   |X^S|\leq \sum_{i=1}^{m}2^{|S|}|S_i|!<\infty.
\]

(2) Let $k$ bound the support sizes. Local oligomorphicity applied with
$S=\emptyset$ gives finitely many $G$-orbits on $A^{j}$ for each
$0\leq j\leq k$, hence finitely many $G$-orbits of subsets of $A$ of
size at most $k$ (send an enumerating tuple to its range); choose
representatives $T_1,\dots,T_r$. Every $x\in X$ has
$\supp(x)=\pi(T_j)$ for some $\pi\in G$ and some $j$, and then
$\pi^{-1}x\in X^{T_j}$ by Lemma~\ref{lem:leastequi}. Therefore
$X=\bigcup_{j\leq r}G\cdot X^{T_j}$, a union of at most
$\sum_j|X^{T_j}|$ orbits, which is finite by uniform finiteness.

For the first counterexample, every element has least support
$\emptyset$, while $X^{\emptyset}=X$ is infinite and every orbit is a
singleton. For the second take $G=\Sym(A)$. An element of
$\coprod_n A^{(n)}$ supported by $S$
is an injective tuple with all entries in $S$, so
$\bigl(\coprod_n A^{(n)}\bigr)^{S}=\coprod_{n\leq|S|}S^{(n)}$ is finite;
supports have size $n$ on the $n$-th summand; and the summands are
$G$-invariant and nonempty, so there are infinitely many orbits.
\end{proof}

\begin{rem}[What the theorem is for]\label{rem:orbitfinitereading}
Theorem~\ref{thm:orbitfinite} places the central finiteness notion of
computing with atoms beside the invariance hierarchy, and separates its
two ingredients. Uniform finiteness is the structural half: it says
that every level $X^S$ of the context filtration is finite. The bound on
support size is a resource constraint, and it is what a register machine
with a fixed number of registers imposes. Dropping it yields a strictly
larger class of state spaces, which can include countable disjoint
unions such as the example above while retaining the finite-level
property needed for symbolic representation one context at a time.
This suggests the most direct computational question raised by the
hierarchy: how far the automata theory of \cite{BKL14} extends to
uniform finite state spaces of unbounded support.
\end{rem}

\subsection{Three context principles: a roadmap}
\label{subsec:threeprinciples}

The meet and freshness results above, together with the join law proved
in Section~\ref{sec:join}, place three context-sensitive principles of
nominal technique in the symmetry parameter. Table~\ref{tab:principles}
presents the correspondence as a roadmap to the remaining sections.

\begin{table}[htbp]
\centering
\caption{Three context-sensitive principles of nominal technique and the
conditions on the data symmetry that carry them.}
\label{tab:principles}
\footnotesize
\begin{tabular}{@{}>{\raggedright\arraybackslash}p{4.3cm}>{\raggedright\arraybackslash}p{4.9cm}>{\raggedright\arraybackslash}p{3.6cm}@{}}
\toprule
Nominal principle & Condition on $G$ & Fails already for\\
\midrule
Canonical registers: least supports & meet law; equivalently
$\Fix_G(S\cap T)=\langle\Fix_G(S)\cup\Fix_G(T)\rangle$
(Thm.~\ref{thm:meetchar}, Cor.~\ref{cor:criterion}) & sessions
$\Aut(A,E)$ (Thm.~\ref{thm:equivfail})\\[2pt]
Freshness quantifier: \textsf{some/any} & $\Fix_G(S)$ transitive on
$A\smallsetminus S$; for closed $G$, only $\Sym(A)$
(Thm.~\ref{thm:someanychar}) & every proper closed symmetry, including
$\Aut(\mathbb Q,<)$\\[2pt]
Invariance decomposes over unions & join law
(Thm.~\ref{thm:join}) & generic $3$-hypergraph
(Thm.~\ref{thm:hypergraph})\\
\bottomrule
\end{tabular}
\end{table}

The table separates three different mechanisms. Meet and join vary
independently, whereas unrestricted some/any is stronger: among closed groups
it forces $G=\Sym(A)$ and hence both lattice laws. The examples in the
next section make the meet/join independence explicit.

\section{The join law and its failure}\label{sec:join}

The meet law asks whether invariance can be pushed \emph{down} to the
intersection of two register sets; the join law asks whether it can be
decomposed \emph{across} their union. The two are formally dual, and
one might expect them to stand or fall together. They do not. This
section shows that relational arity supplies a broad sufficient
condition and natural obstructions for the join law, and exhibits
symmetries realizing each of the four possible combinations.
The mechanism is type-theoretic: the join law holds when the type of a tuple over $S\cup T$ is already determined by its types over $S$ and over $T$ separately, and this determination is automatic when no relation can see a coordinate together with parameters from both sides at once. Under local oligomorphicity and ultrahomogeneity, binary relational
languages therefore satisfy it, while genuinely ternary ones need not.

\begin{thm}[Join law for binary ultrahomogeneous symmetries]
\label{thm:join}
Let $M$ be an ultrahomogeneous relational structure with domain $A$ in a
purely relational language (in particular, with no function symbols) whose
relation symbols have arity at most~$2$, and suppose $G=\Aut(M)$ is locally
oligomorphic. Then for all finite
$S,T\subseteq A$ and~$n\geq1$,
\[
B^{(n)}_{G,S}\vee B^{(n)}_{G,T}\;=\;B^{(n)}_{G,S\cup T}.
\]
\end{thm}
\begin{proof}
By local oligomorphicity, $B^{(n)}_{G,S\cup T}$ is a finite atomic
Boolean algebra whose atoms are the $\Fix_G(S\cup T)$-orbits on $A^n$
(Lemma~\ref{lem:atomic}); it therefore suffices to show that each such
orbit lies in the subalgebra generated by $B^{(n)}_{G,S}\cup
B^{(n)}_{G,T}$. By ultrahomogeneity of $M$, two tuples
$\bar x,\bar y\in A^n$ lie in the same $\Fix_G(U)$-orbit if and only
if they have the same quantifier-free type over $U$, for any finite
$U$: the map fixing $U$ and sending $\bar x$ to $\bar y$ is a partial
isomorphism precisely under type equality, and extends to an
automorphism fixing~$U$ pointwise by ultrahomogeneity. Because the
language is purely relational and binary, the quantifier-free type of
$\bar x$ over $S\cup T$ is
determined by (i) the quantifier-free type of $\bar x$ over
$\emptyset$ (the relations and equalities among the coordinates),
(ii) the binary relations and equalities between each coordinate
$x_i$ and each parameter $c\in S$, and (iii) the same for each $c\in
T$: no atomic formula involves parameters from $S$ and from~$T$
simultaneously together with coordinates, since its arity is at most
$2$. Atomic formulas containing only parameters have a fixed truth
value and therefore do not distinguish the tuple $\bar x$. Items
(i)+(ii) determine the type over $S$ and items (i)+(iii) determine the
type over $T$; hence
\[
\text{orbit of }\bar x\text{ over }S\cup T
\;=\;
\bigl(\text{orbit of }\bar x\text{ over }S\bigr)\cap
\bigl(\text{orbit of }\bar x\text{ over }T\bigr),
\]
an intersection of an element of $B^{(n)}_{G,S}$ with one of
$B^{(n)}_{G,T}$ (orbits over $U$ are $U$-supported). Here local
oligomorphicity is used exactly to ensure that there are only finitely
many $\Fix_G(S\cup T)$-orbits on $A^n$: every
$S\cup T$-supported relation is therefore a \emph{finite} union of
these intersections, as required for membership in the Boolean algebra
generated by the two context levels.
\end{proof}

The binary-language hypothesis is sufficient, not necessary. The operative
condition in the proof is that every tuple type over $S\cup T$ be determined
by its restrictions over $S$ and over $T$. Higher-arity languages may still
satisfy this condition in particular structures, while
Theorems~\ref{thm:hypergraph} and~\ref{thm:betwfail} show that they need
not.

\begin{exa}
Theorem~\ref{thm:join} applies to $\Sym(A)$ (pure equality; any
infinite $A$), to \mbox{$\Aut(\mathbb Q,<)$}, to the random graph, and to
$\Aut(A,E)$ for $(A,E)$ as in Theorem~\ref{thm:equivfail}; that is,
with infinitely many classes, all of the same infinite cardinality,
under which $(A,E)$ is ultrahomogeneous. In particular, the generic
equivalence relation \emph{satisfies} the join law while failing the
meet~law.
\end{exa}

\begin{thm}[Least supports without the join law]
\label{thm:hypergraph}
Let $\mathcal H_3=(A,R)$ be the countable generic irreflexive symmetric
$3$-uniform hypergraph, and let $G=\Aut(\mathcal H_3)$. Then $G$
satisfies the meet law at every arity, but fails the join law already
at arity $1$.
\end{thm}

\begin{proof}
The age of $\mathcal H_3$ has free amalgamation, and $\mathcal H_3$ is
its Fra\"iss\'e limit. For free homogeneous structures,
\cite[Lem.~2.7(i)]{MacphersonTent11} gives the stabilizer identity
$\Fix_G(S\cap T)=\langle\Fix_G(S)\cup\Fix_G(T)\rangle$ for finite
$S,T$. Thus $G$ has the SAP, and Theorem~\ref{thm:meetchar} gives both
least supports and the meet law at every arity. Equivalently, the
model-theoretic route uses weak elimination of imaginaries for
free-amalgamation theories \cite[Thm.~4.7]{Conant17} together with
Proposition~\ref{prop:weakEI}.

Choose distinct $a,b\in A$ and put $U_{a,b}=\{x\in A:R(a,b,x)\}$.
The set $U_{a,b}$ is supported by $\{a,b\}$, so
$U_{a,b}\in B^{(1)}_{G,\{a,b\}}$. The point stabilizer $\Fix_G(a)$
has precisely the two unary orbits $\{a\}$ and $A\setminus\{a\}$:
any isomorphism between two induced two-element substructures fixing~$a$ extends by homogeneity. Thus, $B^{(1)}_{G,\{a\}}=\{\emptyset,\{a\},A\setminus\{a\},A\}$, and similarly for $b$. The Boolean algebra generated by these two
algebras has atoms $\{a\}$, $\{b\}$, and $A\setminus\{a,b\}$. By the
extension property of the generic hypergraph, both $U_{a,b}$ and its
complement meet $A\setminus\{a,b\}$. Hence $U_{a,b}$ is not in the
generated algebra, and
\[
 B^{(1)}_{G,\{a\}}\vee B^{(1)}_{G,\{b\}}
 \subsetneq B^{(1)}_{G,\{a,b\}}. \qedhere
\]
\end{proof}

\begin{thm}[Join-law failure for the betweenness reduct]
\label{thm:betwfail}
Let $B$ be the betweenness relation of $(\mathbb Q,<)$ and
$G=\Aut(\mathbb Q,B)$. Then the join law fails at arity $1$: for any
$a<b$,
\[
B^{(1)}_{G,\{a\}}\vee B^{(1)}_{G,\{b\}}
\subsetneq B^{(1)}_{G,\{a,b\}}.
\]
\end{thm}
\begin{proof}
The group $G$ consists of the order-preserving and order-reversing
bijections of $\mathbb Q$; this is classical, and is the standard
description of the reducts of the dense linear order
\cite{Cameron76}. A reflection fixing
$a$ interchanges the two half-lines, while the order-preserving
stabilizer is transitive on each; hence $\Fix_G(a)$ has orbits
$\{a\}$ and $\mathbb Q\setminus\{a\}$. Therefore the Boolean algebra
generated by the unary algebras at $a$ and $b$ has atoms
$\{a\}$, $\{b\}$, and $\mathbb Q\setminus\{a,b\}$. By contrast, an
order-reversing bijection cannot fix both $a$ and $b$, so
$\Fix_G(\{a,b\})$ is order preserving and leaves the interval $(a,b)$
invariant. Since $(a,b)$ splits
$\mathbb Q\setminus\{a,b\}$, it is not in the generated~algebra.
\end{proof}

\begin{prop}[Meet failure for the betweenness symmetry]
\label{prop:betwmeetfail}
The group $G=\Aut(\mathbb Q,B)$ fails the meet law at arity $2$.
Consequently the betweenness symmetry fails both lattice laws.
\end{prop}

\begin{proof}
Let $<$ denote one of the two linear orders inducing $B$, considered as
a subset of $\mathbb Q^2$. If $a\neq b$, every betweenness
automorphism fixing both points is order preserving, because an
order-reversing bijection of $\mathbb Q$ has at most one fixed point.
Thus every two-element set $\{a,b\}$ supports $<$.

Choose pairwise distinct $a,b,c$. Both $\{a,b\}$ and $\{a,c\}$
support $<$, but their intersection $\{a\}$ does not: an
order-reversing automorphism fixing $a$ maps $<$ to its converse.
Therefore
\[
 {<}\ \in B^{(2)}_{G,\{a,b\}}\cap B^{(2)}_{G,\{a,c\}}
 \setminus B^{(2)}_{G,\{a\}},
\]
which refutes the meet law.
\end{proof}

\begin{rem}[The four global settings]\label{rem:quadrants}
The meet and join laws are independent. All four combinations occur:
\[
\begin{array}{c|c}
(\text{meet},\text{join}) & \text{example}\\ \hline
(\checkmark,\checkmark) & \Sym(A),\ \Aut(\text{random graph})\\
(\times,\checkmark) & \Aut(A,E)\\
(\checkmark,\times) & \Aut(\mathcal H_3)\\
(\times,\times) & \Aut(\mathbb Q,B).
\end{array}
\]
The generic ternary hypergraph shows that failure of the join law is
compatible with canonical least supports. The betweenness reduct shows
that passing to a reduct can destroy both context-composition laws.
\end{rem}

\section{Orientability and context-sensitive asymmetry}\label{sec:orient}

The finite-context hierarchy indicates more than the existence of supports.
It also records relational structure visible after parameters have been
named. This section develops one categorical invariant of that structure.
The motivating contrast is simple. Over ordered atoms the atom object
carries an invariant orientation; over pure equality atoms every candidate
orientation can be reversed by a transposition. Because orientation can be
expressed using products and subobjects, this contrast survives arbitrary
equivalence of the finitely supported action categories.

The invariant will also be relativized to a finite context. This yields a
numerical threshold measuring how many atoms must be named before fresh
atoms can be oriented. The threshold has an exact material interpretation:
it is the least support size of a choice function on unordered pairs.

\subsection{Context levels as subobject lattices of slices}

\begin{prop}[Levels are subobject lattices of slices]\label{prop:slice}
Let $G\leq\Sym(A)$ be closed, let $S\subseteq A$ be finite, and let
$\bar s$ be an injective enumeration of $S$, with orbit
$O_{\bar s}=G\cdot\bar s$. Then
\[
   \FSS_G(A)/O_{\bar s}\simeq \FSS_{\Fix_G(S)}(A).
\]
Under this equivalence the pullback of the atom object $A$ corresponds to
$A$ with the restricted $\Fix_G(S)$-action, and hence
\[
   \Sub(A^n)=\Inv_n(\Fix_G(S))=B^{(n)}_{G,S}\qquad(n\geq1)
\]
in the corresponding slice.
\end{prop}
\begin{proof}
By Corollary~\ref{cor:cont}, $\FSS_G(A)=\mathrm{Cont}(G)$. The orbit
$O_{\bar s}$ is the transitive object $G/\Fix_G(S)$. For any open
subgroup $U\leq G$, the standard induction--fiber equivalence sends an
object $p:X\to G/U$ to the fiber $p^{-1}(U)$ with its $U$-action, and
sends a continuous $U$-set~$Y$ to $G\times_UY\to G/U$. Thus
$\mathrm{Cont}(G)/(G/U)\simeq\mathrm{Cont}(U)$. Taking
$U=\Fix_G(S)$ gives the first assertion. Under the equivalence,
$A\times O_{\bar s}\to O_{\bar s}$ corresponds to $A$ with the
restricted $U$-action. Lemma~\ref{lem:subobj} then identifies subobjects
of $A^n$ with the $U$-invariant subsets of $A^n$, namely
$B^{(n)}_{G,S}$.
\end{proof}

Thus every context level has an intrinsic categorical interpretation once the corresponding tuple orbit is distinguished. Contextual invariants below are therefore properties of these slices, not merely of a chosen group presentation.

\subsection{Orientations and the swap criterion}

\begin{defi}\label{def:orient}
Let $\mathcal C$ be a category with finite limits in which the relevant
subobject posets have finite joins, and let $Q$ be an object. Write
$\Delta_Q\leq Q\times Q$ for the diagonal and
$\sigma=\langle\mathrm{pr}_2,\mathrm{pr}_1\rangle$ for factor exchange.
An \emph{orientation} of $Q$ is a subobject $O\leq Q\times Q$ satisfying
\[
 O\wedge\sigma(O)=\bot,
 \qquad
 O\vee\sigma(O)\vee\Delta_Q=\top.
\]
An atom $Q$ is \emph{properly orientable} if it admits an orientation and
$\Delta_Q\neq\top$.
\end{defi}

In $\FSS_G(A)$, a properly orientable atom is exactly a non-singleton
transitive $G$-set carrying a $G$-invariant tournament.

\begin{lem}\label{lem:orientinv}
Equivalences preserve atoms, products, diagonals, factor exchange and
subobject lattices. Consequently the existence of a properly orientable
atom is invariant under equivalence of categories.
\end{lem}
\begin{proof}
An equivalence preserves finite limits up to canonical isomorphism and
induces order isomorphisms on corresponding subobject posets. Hence it
preserves the two equations of Definition~\ref{def:orient}, as well as the
property of an object being an atom and the condition
$\Delta_Q\neq\top$.
\end{proof}

\begin{lem}[Swap dichotomy]\label{lem:swap}
Let $Q$ be a transitive $G$-set and let $O$ be a $G$-orbit in
$Q\times Q$ disjoint from the diagonal. Then $\sigma(O)$ is again an
orbit, and $\sigma(O)=O$ if and only if some $g\in G$ exchanges two
distinct elements forming a pair in $O$.
\end{lem}
\begin{proof}
The map $\sigma$ commutes with the diagonal $G$-action. If
$g(u,v)=(v,u)$ with $(u,v)\in O$, then $O$ and $\sigma(O)$ meet and hence
are equal. Conversely, if $\sigma(O)=O$, then for $(u,v)\in O$ the pair
$(v,u)$ lies in the same orbit, so some $g\in G$ sends $(u,v)$ to
$(v,u)$.
\end{proof}

\begin{thm}[Tournament criterion]\label{thm:orientcrit}
Let $Q\cong G/U$ be an atom of $\FSS_G(A)$ with at least two elements.
The following are equivalent.
\begin{enumerate}[label=\textup{(\roman*)}]
\item $Q$ is properly orientable;
\item no element of $G$ exchanges two distinct elements of $Q$;
\item $gU\cap Ug^{-1}=\emptyset$ for every $g\in G\setminus U$.
\end{enumerate}
\end{thm}
\begin{proof}
(i)$\Rightarrow$(ii). If $g$ exchanges $u\neq v$, the orbital of
$(u,v)$ is self-paired by Lemma~\ref{lem:swap}. Any invariant
orientation is a union of orbitals, so totality forces this orbital into
one of the two directions while self-pairing forces it into both,
contradicting asymmetry.

(ii)$\Rightarrow$(i). Factor exchange acts on the off-diagonal orbitals
as a fixed-point-free involution. Using external choice in the ambient
ZFC metatheory, choose one orbital from each pair
$\{O,\sigma(O)\}$ and take their union. This invariant relation is total
and asymmetric off the diagonal.

(ii)$\Leftrightarrow$(iii). It is enough to test whether some element
exchanges the cosets $U$ and $gU$, $g\notin U$. Such an element $h$
exists exactly when $h\in gU$ and $hg\in U$, equivalently
$h\in gU\cap Ug^{-1}$.
\end{proof}

The criterion has a standard permutation-group reformulation. Orbitals of
$G$ on $Q\times Q$ correspond to double cosets $U\backslash G/U$, and
factor exchange is the usual pairing involution on orbitals
\cite[Sec.~1.11]{Cameron99}.

\begin{thm}[Double-coset form]\label{thm:doublecoset}
For a non-singleton atom $Q\cong G/U$, the conditions of
Theorem~\ref{thm:orientcrit} are further equivalent to:
\begin{enumerate}[label=\textup{(\roman*)},start=4]
\item no off-diagonal orbital of $(G,Q)$ is self-paired;
\item $g^{-1}\notin UgU$ for every $g\in G\setminus U$.
\end{enumerate}
\end{thm}
\begin{proof}
Condition (iv) is Lemma~\ref{lem:swap}. Moreover
$gU\cap Ug^{-1}\neq\emptyset$ iff $gu=vg^{-1}$ for some $u,v\in U$,
which is equivalent to $g^{-1}=v^{-1}gu\in UgU$. Therefore
$gU\cap Ug^{-1}=\emptyset$ iff $g^{-1}\notin UgU$. Under the
orbital--double-coset correspondence this is precisely the absence of a
self-paired off-diagonal orbital.
\end{proof}

\begin{cor}[Rank parity]\label{cor:rankodd}
If a properly orientable atom $Q$ has finite permutation rank~$r$, then
$r$ is odd. In particular, a non-singleton $2$-transitive atom is not
properly orientable.
\end{cor}
\begin{proof}
Pairing is fixed-point-free on the $r-1$ off-diagonal orbitals, so
$r-1$ is even.
\end{proof}

\begin{thm}[Involution obstruction]\label{thm:involution}
Let $Q$ be an atom of $\FSS_G(A)$.
\begin{enumerate}[label=\textup{(\alph*)}]
\item If an involution $t\in G$ acts nontrivially on $Q$, then $Q$ is
not properly orientable.
\item If the subgroup generated by the involutions of $G$ is dense in
$G$ for pointwise convergence, then no non-singleton atom of
$\FSS_G(A)$ is properly orientable.
\end{enumerate}
\end{thm}
\begin{proof}
For (a), if $tq\neq q$ then $t$ exchanges $q$ and $tq$, so
Theorem~\ref{thm:orientcrit} applies. For (b), let~$N$ be the kernel of
the action on a properly orientable atom $Q$. Each point stabilizer of
$Q$ is open, hence closed, so $N$ is closed. Part (a) puts every
involution in $N$; density of the subgroup they generate gives $N=G$,
forcing $Q$ to be a singleton.
\end{proof}

For example, transpositions generate the dense subgroup
$\Symfin(A)\leq\Sym(A)$, so Theorem~\ref{thm:involution} already proves
that equality-atom universes have no properly orientable atom at any
infinite cardinality.

\subsection{Free amalgamation and homogeneous atom structures}

\begin{defi}\label{def:flipamalg}
A class $\mathcal A$ of finite relational structures has
\emph{symmetric self-amalgamation over $\emptyset$} if for every
$c\in\mathcal A$ there is $W\in\mathcal A$ consisting of two disjoint
copies $c_1,c_2$ of $c$ and carrying an automorphism of order two that
exchanges them. Free amalgamation over $\emptyset$ implies this property.
\end{defi}

\begin{thm}[Non-orientability from symmetric self-amalgamation]
\label{thm:freenonorient}
Let $M$ be an ultrahomogeneous relational structure of any infinite
cardinality, put $G=\Aut(M)$, and suppose that its age has symmetric
self-amalgamation over $\emptyset$ and that $G$ has the least-support
property. Then no atom of $\FSS_G(M)$ is properly orientable.
\end{thm}
\begin{proof}
Let $Q$ be a non-singleton atom and $q\in Q$, with nonempty least support
$S$. Enumerate $S$ by $\bar a$. By the extension property presented in Section~\ref{sec:regime}, embed over the given copy of $S$ a symmetric
self-amalgam of its induced structure; write $\bar a'$ for the disjoint
mirror copy, and let $p$ be the partial isomorphism exchanging the two
copies. By ultrahomogeneity, extend $p$ to $g\in G$. Since $g$ agrees
with the involution $p$ on the whole finite domain,
$g^2\in\Fix_G(S\cup S')$, and hence $g^2q=q$. Equivariance of least
supports gives $\supp(gq)=gS=S'$, which is disjoint from the nonempty set
$S$; therefore $gq\neq q$. Thus $g$ exchanges $q$ and $gq$, contradicting
Theorem~\ref{thm:orientcrit}. Singleton atoms are excluded by the
requirement $\Delta_Q\neq\top$.
\end{proof}

\begin{cor}\label{thm:nonorient}
For every infinite $A$, no atom of $\FSS_{\Sym(A)}(A)$ is properly
orientable. The same holds for the finitely supported universes of the
countable random graph and the Henson graphs.
\end{cor}
\begin{proof}
The equality case also follows directly from
Theorem~\ref{thm:involution}. For the random graph and Henson graphs, the
ages have free amalgamation, and
\cite[Lem.~2.7(i)]{MacphersonTent11} gives
$\Fix_G(S\cap T)=\langle\Fix_G(S)\cup\Fix_G(T)\rangle$ for finite
$S,T$; hence least supports follow from Theorem~\ref{thm:meetchar}.
Apply Theorem~\ref{thm:freenonorient}.
\end{proof}

\noindent
For the canonical atom object one can test orientability from two-point substructures~alone.

\begin{thm}[Two-point criterion]\label{thm:homogcrit}
Let $M$ be an ultrahomogeneous relational structure with infinite domain
$A$, and suppose $G=\Aut(M)$ is transitive on $A$. Then the following are
equivalent:
\begin{enumerate}[label=\textup{(\roman*)}]
\item the atom object $A$ is properly orientable;
\item for all distinct $a,b\in A$, the ordered pairs $(a,b)$ and $(b,a)$
lie in different $G$-orbits;
\item no two-element induced substructure of $M$ admits the nonidentity
permutation of its domain as an automorphism;
\item $G$ preserves a tournament on $A$.
\end{enumerate}
If $M$ is moreover $\omega$-categorical, these are equivalent to the
existence of a parameter-free definable tournament on $M$.
\end{thm}
\begin{proof}
By ultrahomogeneity, the transposition of $\{a,b\}$ extends to an
automorphism exactly when it is an automorphism of the induced
substructure on $\{a,b\}$, giving (ii)$\Leftrightarrow$(iii).
Theorem~\ref{thm:orientcrit} applied to $Q=A$ gives
(i)$\Leftrightarrow$(ii), while the translation of
Definition~\ref{def:orient} gives (i)$\Leftrightarrow$(iv). In the
$\omega$-categorical case, invariant relations on finite powers are exactly
the parameter-free definable relations by Ryll--Nardzewski.
\end{proof}

\begin{cor}[Symmetric irreflexive languages]\label{cor:symmlang}
Assume the hypotheses of Theorem~\ref{thm:homogcrit}. Suppose in addition
that every basic relation is invariant under permutations of its
coordinates and, for arity at least two, is true only on tuples of pairwise
distinct elements. Then the canonical atom object $A$ is not properly
orientable.
\end{cor}
\begin{proof}
On a two-element induced substructure, relations of arity at least three
are empty; binary relations are symmetric; and transitivity makes every
unary relation either empty or universal. Hence the transposition of the
two points is an automorphism, so Theorem~\ref{thm:homogcrit}(iii) fails.
\end{proof}

\subsection{Ordered symmetries and categorical inequivalence}

\begin{thm}[Orientable universes]\label{thm:orientable}
Let $G\leq\Sym(A)$ preserve a tournament $T$ on $A$, and let
$Q\subseteq A$ be any $G$-orbit with at least two elements. Then $Q$ is
a properly orientable atom, with orientation $T\cap(Q\times Q)$.
\end{thm}
\begin{proof}
The orbit $Q$ is an atom by Lemma~\ref{lem:subobj}, and
$T\cap Q^2$ is invariant, asymmetric and total off the diagonal.
\end{proof}

\begin{cor}[Categorical inequivalence]\label{cor:noneq}
If $G\leq\Sym(A)$ preserves a tournament on $A$ and $A$ has a
non-singleton $G$-orbit, then $\FSS_G(A)$ is equivalent to no equality-atom
universe $\FSS_{\Sym(B)}(B)$ with $B$ infinite; in particular it is not
equivalent to $\Nom$.
\end{cor}
\begin{proof}
Use Theorem~\ref{thm:orientable}, Corollary~\ref{thm:nonorient}, and
Lemma~\ref{lem:orientinv}.
\end{proof}

Orientability of one atom is a coarse invariant. When the symmetry itself
preserves a linear order, much more is true.

\begin{defi}\label{defi:orderable}
An atom is \emph{orderable} if it carries an invariant linear order. A
category is \emph{totally orderable} if every non-singleton atom is
orderable.
\end{defi}

\begin{thm}[Ordered symmetries are totally orderable]
\label{thm:totalorder}
Let $G\leq\Sym(A)$ preserve a linear order on $A$ and have the
least-support property. Then:
\begin{enumerate}[label=\textup{(\arabic*)}]
\item $\mathrm{Stab}_G(x)=\Fix_G(\supp x)$ for every finitely supported
$x$;
\item every atom of $\FSS_G(A)$ is isomorphic to a $G$-orbit of a finite
subset of $A$, and every non-singleton atom is orderable;
\item every object of $\FSS_G(A)$ admits a $G$-invariant linear order.
\end{enumerate}
\end{thm}
\begin{proof}
If $g$ fixes $x$, equivariance of least supports gives
$g(\supp x)=\supp x$. Since $g$ preserves the ambient linear order, its
restriction to the finite ordered set $\supp x$ is the identity. Hence
$g\in\Fix_G(\supp x)$, proving (1); the reverse inclusion is the definition
of support.

For (2), the map $gx\mapsto g(\supp x)$ identifies the orbit of $x$ with
the orbit of its finite least support. Order each orbit of $n$-element
subsets by lexicographically comparing increasing enumerations. The order
is $G$-invariant. For (3), externally well-order the set of $G$-orbits of
an object and concatenate the invariant linear orders on the individual
orbits. External choice is available in the surrounding ZFC metatheory.
\end{proof}

For $G=\Aut(\mathbb Q,<)$, least supports hold
\cite[Cor.~9.5]{BKL14}; hence all atoms in
$\FSS_G(\mathbb Q)$ are orderable. This strengthens the categorical
separation from equality atoms: the ordered universe remembers an invariant
order on every one of its atoms, not merely on the canonical atom object.

\subsection{Contextual orientability and pair choice}

\begin{defi}[Contextual orientability]\label{defi:ctxorient}
For finite $S\subseteq A$, write $\mathrm{or}_G(S)$ when
$\Fix_G(S)$ preserves a tournament on $A\setminus S$. Define the
\emph{orientability threshold}
\[
   \tau(G)=\min\{|S|: S\subseteq A\text{ finite and }\mathrm{or}_G(S)\}
   \in\mathbb N\cup\{\infty\},
\]
with value $\infty$ if no finite context orients the fresh atoms.
\end{defi}

The quantity $\tau(G)$ is a \emph{pointed} invariant: its definition uses
the chosen atom object~$A$ and the context slices determined by finite tuples
of that object. Proposition~\ref{prop:slice} makes each fixed level
categorical once this data is distinguished; it does not imply that
$\tau(G)$ can be recovered from the bare category $\FSS_G(A)$ without a
distinguished atom object. That reconstruction problem is left open in
Question~\ref{q:reconstruction}.

By Proposition~\ref{prop:slice}, $\mathrm{or}_G(S)$ is a property of the
atom object in the slice corresponding to the tuple $S$, after deleting
the points fixed by that context.

\begin{thm}[Contextual orientability is sortwise orientability]
\label{thm:ctxsorts}
For finite $S\subseteq A$, $\mathrm{or}_G(S)$ holds if and only if every
freshness sort over $S$ carries a $\Fix_G(S)$-invariant tournament.
Moreover, $S\subseteq T$ and $\mathrm{or}_G(S)$ imply
$\mathrm{or}_G(T)$.
\end{thm}
\begin{proof}
A tournament on $A\setminus S$ restricts to each freshness sort. Conversely,
orient each sort invariantly, externally well-order the set of sorts, and
orient pairs from different sorts according to that external order. Every
element of $\Fix_G(S)$ preserves each sort setwise, so the resulting
tournament is invariant. Monotonicity follows because
$\Fix_G(T)\leq\Fix_G(S)$ and restriction of a tournament on
$A\setminus S$ gives one on $A\setminus T$.
\end{proof}

\begin{thm}[Orientability threshold as support cost of pair choice]
\label{thm:pairchoice}
Let $S\subseteq A$ be finite. The following are equivalent:
\begin{enumerate}[label=\textup{(\roman*)}]
\item $[A]^2$ admits an $S$-supported choice function
$c:[A]^2\to A$ with $c(P)\in P$;
\item $\mathrm{or}_G(S)$.
\end{enumerate}
Consequently $[A]^2$ admits a finitely supported choice function if and
only if $\tau(G)<\infty$, and in that case $\tau(G)$ is the least
cardinality of a finite support of such a function.
\end{thm}
\begin{proof}
If $c$ is $S$-supported, orient distinct $x,y\notin S$ by
$x\to y$ iff $c(\{x,y\})=x$. Equivariance under $\Fix_G(S)$ makes this
a $\Fix_G(S)$-invariant tournament.

Conversely, let $T$ be a $\Fix_G(S)$-invariant tournament on
$A\setminus S$ and enumerate $S$ externally as $s_1,\ldots,s_k$. For an
unordered pair choose the winner of $T$ when both entries lie outside
$S$, choose the entry in $S$ when exactly one lies there, and choose the
entry with smaller index when both lie in $S$. Every element of
$\Fix_G(S)$ fixes the enumeration pointwise and preserves $T$, so the
resulting choice function is $S$-supported. Minimizing over finite supports
gives the last assertion.
\end{proof}

For equality atoms, $\tau(\Sym(A))=\infty$ by the transposition argument of
Theorem~\ref{thm:material}(b). The standard reducts of the dense order give
a finite hierarchy of nontrivial thresholds. Let $C$, $B$ and $D$ denote,
respectively, the usual circular-order, betweenness and separation reducts
of $(\mathbb Q,<)$ \cite{Cameron76}.

\begin{thm}[Thresholds for the standard order reducts]\label{thm:threshold}
\[
\tau(\Aut(\mathbb Q,<))=0,\qquad
\tau(\Aut(\mathbb Q,C))=1,\qquad
\tau(\Aut(\mathbb Q,B))=2,\qquad
\tau(\Aut(\mathbb Q,D))=3,
\]
while $\tau(\Sym(\mathbb Q))=\infty$. Hence the threshold takes the
five values $0,1,2,3,\infty$ on the five closed reduct groups of the
dense order.
\end{thm}
\begin{proof}
The order $<$ itself gives threshold $0$. The circular-order group is
$2$-transitive, so it preserves no tournament without parameters by
Corollary~\ref{cor:rankodd}; fixing one point cuts the circle into a dense
linear order, giving threshold $1$.

For betweenness, an order-reversing automorphism fixing one point is an
involution that exchanges points of the complement, so one parameter does
not suffice. An order-reversing automorphism of $(\mathbb Q,<)$ has at
most one fixed point; hence the pointwise stabilizer of two distinct points
inside $\Aut(\mathbb Q,B)$ is order-preserving, and two parameters do
suffice.

For separation, fixing one point makes the induced structure on the
complement a betweenness structure (equivalently, the point stabilizer is
the corresponding betweenness group; see \cite[Lem.~2.4(d)]{JunkerZiegler08}
and the reduct classification of \cite{Cameron76}). The previous paragraph
therefore shifts the threshold by one, giving $3$. The value for the full symmetric group is
the pair-choice obstruction above.
\end{proof}

\begin{rem}[One mechanism, two interpretations]\label{rem:kinship}
Theorem~\ref{thm:pairchoice} identifies a categorical-looking asymmetry
with a material support cost. A context orients the fresh atoms exactly
when the same context supports a choice from every unordered pair. Thus
the classical Fraenkel--Mostowski transposition obstruction and the
categorical absence of an orientation are two forms of the same finite
dependence phenomenon.
\end{rem}

\begin{rem}[Limits of orientability]\label{rem:rado}
Orientability is deliberately coarse. Equality atoms, the random graph
and other free-amalgamation examples are all non-orientable, so this
invariant does not separate their action categories. The threshold is a pointed invariant built from the slices of Proposition~\ref{prop:slice};
whether comparable data can be recovered from the bare category without a
distinguished atom object belongs to the reconstruction problem of
Question~\ref{q:reconstruction}. Finer invariants can be obtained by
asking which finite relational structures occur invariantly on atoms and
their finite powers (Question~\ref{q:spectrum}).
\end{rem}

\section{Model-theoretic characterizations}\label{sec:modeltheory}

The witnesses in Theorems~\ref{thm:equivfail} and~\ref{thm:betwfail}
are not arbitrary: the meet law failed on an \emph{imaginary} (an
equivalence class definable from any of its members but from no
smaller real parameter set), and the join law failed on a relation
whose definition irreducibly couples parameters from both supports
through a relation symbol of arity $3$ (betweenness $B(a,x,b)$). The
results below make these correspondences precise in an important
$\omega$-categorical framework. We use weak elimination of
imaginaries in its standard finite-real-code form: every
$e\in M^{\mathrm{eq}}$ has a finite real tuple $\bar c$ with
$e\in\mathrm{dcl}^{\mathrm{eq}}(\bar c)$ and
$\bar c\in\mathrm{acl}^{\mathrm{eq}}(e)$.

\begin{prop}[Meet law and weak elimination of imaginaries]
\label{prop:weakEI}
Let $M$ be a countable $\omega$-categorical structure with degenerate
algebraic closure, and put $G=\Aut(M)$. Then the following are
equivalent:
\begin{enumerate}[label=\textup{(\roman*)}]
\item $M$ has weak elimination of imaginaries;
\item every finitely supported element of every $G$-set has a least
support;
\item the meet law holds at every arity;
\item the stabilizer-amalgamation property of
Definition~\ref{def:sap} holds (equivalently, so does its Galois form).
\end{enumerate}
\end{prop}
\begin{proof}
Assume first that $M$ has weak elimination of imaginaries. By
Theorem~\ref{thm:meetchar}, the least-support property is equivalent
to the meet law at every arity, so it suffices to fix $n$ and finite
$S,T$ and show that any $R\subseteq M^n$ invariant under both
$\Fix_G(S)$ and $\Fix_G(T)$ is invariant under $\Fix_G(S\cap T)$.

By the Ryll--Nardzewski theorem (see, e.g., \cite{Hodges93}), over any finite parameter set
$U\subseteq M$ there are finitely many $n$-types, each isolated by a
formula over $U$; a subset of $M^n$ is $\Fix_G(U)$-invariant iff it is
a union of $U$-definable types, i.e.\ iff it is definable over $U$.
Thus our hypotheses say $R$ is definable over $S$ and definable over
$T$.

Let $e=\ulcorner R\urcorner\in M^{\mathrm{eq}}$ be a canonical
parameter of the definable set $R$. The standard canonical-parameter
property says that, for a real or imaginary parameter set $C$, the set
$R$ is definable over $C$ exactly when
$e\in\mathrm{dcl}^{\mathrm{eq}}(C)$.
Definability of $R$ over $S$ gives
$e\in\mathrm{dcl}^{\mathrm{eq}}(S)$, and over $T$ gives
$e\in\mathrm{dcl}^{\mathrm{eq}}(T)$. By weak elimination of
imaginaries there is a finite real tuple
$\bar c\in M^{<\omega}$ with
\[
e\in\mathrm{dcl}^{\mathrm{eq}}(\bar c)
\qquad\text{and}\qquad
\bar c\in\mathrm{acl}^{\mathrm{eq}}(e).
\]
From $e\in\mathrm{dcl}^{\mathrm{eq}}(S)$ and
$\bar c\in\mathrm{acl}^{\mathrm{eq}}(e)$ we get
$\bar c\in\mathrm{acl}^{\mathrm{eq}}(S)$; intersecting with the real
sort, $\bar c\in\mathrm{acl}(S)=S$ by degeneracy of algebraic closure.
Likewise $\bar c\in\mathrm{acl}(T)=T$. Hence the entries of $\bar c$
lie in $S\cap T$, and $e\in\mathrm{dcl}^{\mathrm{eq}}(\bar c)
\subseteq\mathrm{dcl}^{\mathrm{eq}}(S\cap T)$, so $R$ is definable
over $S\cap T$, i.e.\ $\Fix_G(S\cap T)$-invariant. (Only weak
elimination of imaginaries is used: we never need $e$ itself to be
interdefinable with a real tuple, only to be coded up to algebraic
closure, which is what lets the real ``shadow'' $\bar c$ be
forced into $S$ and into $T$ separately.)

Conversely, assume the meet law, hence the least-support property by
Theorem~\ref{thm:meetchar}, and let $e\in M^{\mathrm{eq}}$ be an
imaginary. The orbit $G\cdot e$ is a $G$-set all of whose elements
are finitely supported. Choose a finite real tuple
$\bar a$ with $e\in\mathrm{dcl}^{\mathrm{eq}}(\bar a)$; its entries
support $e$. Let~$S$ be the least support of $e$. Then
$e\in\mathrm{dcl}^{\mathrm{eq}}(S)$. For every
$g\in\mathrm{Stab}_G(e)$, the set $gS$ also supports~$e$; minimality
and finiteness imply $gS=S$. Hence the orbit of each $s\in S$ under
$\mathrm{Stab}_G(e)$ is finite, so, by the standard
stabilizer--orbit characterization of algebraicity in $M^{\mathrm{eq}}$
(see, e.g., \cite{Hodges93}),
$s\in\mathrm{acl}^{\mathrm{eq}}(e)$. Thus,
\[
 e\in\mathrm{dcl}^{\mathrm{eq}}(S)
 \qquad\text{and}\qquad
 S\subseteq\mathrm{acl}^{\mathrm{eq}}(e),
\]
which is weak elimination of imaginaries. This proves
(i)$\Leftrightarrow$(ii); Theorem~\ref{thm:meetchar} and
Corollary~\ref{cor:criterion} supply the equivalences with (iii) and
(iv).
\end{proof}

\begin{rem}[Provenance and role of the correspondence]\label{rem:weakEIcredit}
The least-support criterion is due to Boja\'nczyk, Klin and Lasota
\cite{BKL14}, and the background theory of weak elimination of imaginaries
is standard; see, for example, \cite{CasanovasFarre}. The contribution used
here is the relational packaging supplied by Theorem~\ref{thm:meetchar}: in the $\omega$-categorical no-algebraicity class, the known least-support condition, and hence weak elimination of imaginaries by Proposition~\ref{prop:weakEI}, becomes a \emph{lattice law} of the invariance hierarchy $S\mapsto B^{(n)}_{G,S}$. The independence results (Theorems~\ref{thm:equivfail},~\ref{thm:betwfail}, Proposition~\ref{prop:betwmeetfail}) then locate the meet and join laws in distinct model-theoretic phenomena. Under these correspondences the failure of the meet law for the generic equivalence relation (Theorem~\ref{thm:equivfail}) is exactly its well-known failure of weak elimination of imaginaries, its $E$-classes being the paradigmatic non-eliminable imaginaries.

Two structural points reinforce these correspondences. First, our
stabilizer amalgamation property (Definition~\ref{def:sap}),
$\Fix_G(S\cap T)=\langle\Fix_G(S)\cup\Fix_G(T)\rangle$, is precisely
the permutation-group formulation of weak elimination of imaginaries
used by Paolini \cite{Paolini24}, when finite parameter sets $S,T$ are replaced by finite
algebraically closed parameter sets; for $\Sym(A)$ and the other
no-algebraicity examples the two coincide, and
Theorem~\ref{thm:meetchar} is the assertion that this amalgamation
condition is equivalent to the meet law on supported~relations.
\end{rem}

We close the section by assembling the preceding correspondences into
structural correspondences and tests for the invariance hierarchy.
Their methodological role is the following: three features of the finitely supported universe $\FSS_G(M)$ (two lattice-theoretic, one categorical) each of which is formulated without reference to a chosen syntax for $M$, constrain a different model-theoretic feature of the symmetry.

\begin{rem}[Structural dictionary]\label{prop:dictionary}
Let $M$ be a countable $\omega$-categorical ultrahomogeneous relational
structure with degenerate algebraic closure and $G=\Aut(M)$. Then:
\begin{enumerate}[label=\textup{(\alph*)}]
\item \textup{(meet $\leftrightarrow$ imaginaries)} the following are
equivalent: $M$ has weak elimination of imaginaries; $G$ has the
least-support property; and the level map
$S\mapsto B^{(n)}_{G,S}$ preserves binary meets at every arity
\textup{(Theorem~\ref{thm:meetchar}, Proposition~\ref{prop:weakEI})}.
The generic equivalence relation illustrates their simultaneous
failure: its equivalence classes are non-eliminable imaginaries and
explicit witnesses to the failure of the meet law
\textup{(Theorem~\ref{thm:equivfail})};
\item \textup{(join and relational arity)} if the language of $M$
is binary, the level map preserves binary joins at every arity
\textup{(Theorem~\ref{thm:join})}; genuinely ternary coupling can
obstruct it, as witnessed by the generic $3$-uniform hypergraph and
the betweenness reduct
\textup{(Theorems~\ref{thm:hypergraph} and~\ref{thm:betwfail})};
\item \textup{(orientability $\leftrightarrow$ tournaments)} an atom
$G/U$ of $\FSS_G(M)$ admits an invariant orientation if and only if
its orbit carries a $G$-invariant tournament, equivalently
$gU\cap Ug^{-1}=\emptyset$ for all $g\notin U$
\textup{(Theorem~\ref{thm:orientcrit})}; free amalgamation with least
supports forbids this \textup{(Theorem~\ref{thm:freenonorient})},
while an invariant linear order or tournament produces it
\textup{(Theorem~\ref{thm:orientable})}.
\end{enumerate}
For the pure set, the random graph and the Henson graphs, the three tests give
\[
   (\text{meet},\text{join},\text{orientability})
   =(\checkmark,\checkmark,\times).
\]
These tests therefore do not separate the two free-amalgamation examples from the nominal universe. By contrast, $\Aut(\mathbb Q,<)$ and the automorphism group of the generic tournament are orientable and hence categorically separated from it.
\end{rem}

\begin{rem}[The method, and its ceiling]\label{rem:method}
The structural dictionary above relates structural and categorical
features of $\FSS_G(M)$ to the model theory of $M$. Within the $\omega$-categorical, degenerate-algebraic-closure framework, the meet law characterizes weak elimination of imaginaries.
Beyond that framework, the appropriate replacement by algebraically closed
supports remains open (Question~\ref{q:charmeet}). A separate
limitation is that non-orientability does not separate
the random graph from the pure set (Remark~\ref{rem:rado}). Completing this family of invariants (finding a family of invariants of $\FSS_G(M)$ that recovers $G$ up to the appropriate equivalence) is precisely the reconstruction and spectrum programme of
Questions~\ref{q:reconstruction} and~\ref{q:spectrum}. The contribution of this paper is to show that the finitely supported universe carries genuine structural information about the symmetry, fine enough to separate the ordered world from the nominal one outright.
\end{rem}

\section{Modelling guidance for computing with atoms}\label{sec:modelling}

The preceding sections developed the hierarchy as mathematics. We now
translate the results into modelling choices for computation over atoms, in
the sense of the automata-over-atoms programme \cite{BKL14} and nominal
semantics \cite{Pitts13}. This section introduces no new proof dependency:
each recommendation is a reading of a numbered result above, with its
hypotheses retained.
Throughout, a \emph{data symmetry} is a pair $(A,G)$ with $G\leq\Sym(A)$: the atoms $A$ model the values a program can store and compare, and $G$ models the structure on values that programs may observe.
The categorical finite-support semantics used here is compatible with the nominal-$G$-set viewpoint of \cite{BKL14}: both study sets with actions whose relevant elements have finite support. The FSS perspective used in this paper keeps that action semantics together with the material ZFA presentation and allows the atom cardinality and the symmetry parameter to be varied independently. For ordered, graph-structured, or session-partitioned alphabets, this common finite-support semantics makes the dependence on the observable data structure explicit before a particular automaton model is chosen. Nothing in this section is needed for the proofs elsewhere; conversely, each claim below is backed by a numbered result of the preceding~sections.

\subsection{Registers and canonical contents}
\label{subsec:registers}

In a register machine over a data symmetry, a configuration stores
finitely many values, and the semantics of a state depends only on
those values: semantically, a set of register contents is a
\emph{finite support} of the residual behaviour. The literature routinely refers to ``the'' registers of a configuration, tacitly presupposing a canonical minimal choice, that is a \emph{least} support. Boja\'nczyk, Klin and Lasota treat this least-support property for arbitrary data symmetries \cite[Def.~4.11, Fact~9.2, Thm.~9.3]{BKL14}. In the countable $\omega$-categorical no-algebraicity setting, Proposition~\ref{prop:weakEI} identifies it additionally with weak elimination of imaginaries.

Theorem~\ref{thm:meetchar} does not enlarge the scope of the BKL criterion; it places that criterion inside the invariance hierarchy. Thus canonical register contents exist exactly when the meet law holds for the hierarchy, equivalently when the associated stabilizer-amalgamation law holds. Corollary~\ref{cor:criterion}
presents the concrete group identity
$\Fix_G(S\cap T)=\langle\Fix_G(S)\cup\Fix_G(T)\rangle$.
Thus the criterion is a property of the symmetry alone, checkable
before any machine model is fixed. It licenses canonical register
contents and intersection-based support reasoning. Independence from
the choice of a fresh witness is a separate issue, governed by the
freshness sorts of Section~\ref{subsec:someany}.

\subsection{Session atoms and the failure of canonical registers}
\label{subsec:sessions}

Let $(A,E)$ be the generic equivalence relation, namely, infinitely many classes, all of the same infinite cardinality
(Theorem~\ref{thm:equivfail} and Remark~\ref{rem:whyequinum}), and
$G=\Aut(A,E)$. This data symmetry models
values that carry a \emph{session}, shard, or group identity which
programs can test (``same session?'')\ but not name: database values
under \textsf{group-by}, protocol messages under session
identifiers, keys under consistent hashing.

Theorem~\ref{thm:equivfail} shows that the meet law fails already for
unary relations. The witness is the session itself. If $aEb$, then the
class $[a]_E$ is supported by $\{a\}$ and by $\{b\}$, but not by their
intersection $\emptyset$. A residual behaviour that depends only on the
session therefore has several incomparable minimal supports and no least one.

This failure does not invalidate the abstract Myhill--Nerode theorem or
deterministic minimization: the relevant results of Boja\'nczyk--Klin--Lasota
\cite[Thm.~3.8 and Thm.~5.2]{BKL14} do not assume least supports. What is
lost is a canonical support-based representation of a state or residual
behaviour. An implementation must either retain a representative and prove
independence from that choice, or change the presented symmetry, for example
by adding a sort of session classes. The hierarchy does not prescribe which
remedy to use; it identifies the boundary between the regime with canonical
supports and the regime without them.

\subsection{Ordered and equality atoms have inequivalent semantic universes}
\label{subsec:ordered}

Timestamps, priorities, versions: the data symmetry is
$\Aut(\mathbb Q,<)$. The order itself is an invariant tournament on
the atom orbit, so the universe is orientable
(Theorem~\ref{thm:orientable}), and
Corollary~\ref{cor:noneq} applies: $\FSS_{\Aut(\mathbb Q,<)}(\mathbb
Q)$ is not equivalent to the standard nominal-set category $\Nom$, nor to any
pure equality-atom universe $\FSS_{\Sym(B)}(B)$.

Corollary~\ref{cor:noneq} rules out a lossless and essentially
surjective translation at the level of the entire finite-support semantic
universe: an equivalence would preserve the existence of a properly
orientable atom, while equality-atom universes have none. This does not rule
out faithful encodings of selected objects, simulations of particular
automaton models, or translations that deliberately forget or add structure.
Expressiveness separations between ordered-atom and equality-atom machines
are known at the automata level \cite{BKL14}; the conclusion here is
categorical inequivalence, not non-encodability in every computational
sense.
The converse does not follow: the invariant does not separate the random-graph universe from the nominal one (Remark~\ref{rem:rado}), so graph-structured data might still admit a lossless categorical reduction to pure names; Question~\ref{q:spectrum} makes the needed finer invariants~precise.

The two remaining symmetries of our four-quadrant analysis are settled on their atom orbits by an elementary swap argument, which we record because it informs the model-selection criteria below.

\begin{rem}[Swap obstruction on the atom orbit]
\label{rem:swaporbit}
Let $G\leq\Sym(A)$ act transitively on~$A$, and suppose that for some
$a\neq b$ in $A$ the transposition-like situation occurs: some $g\in
G$ satisfies $g(a)=b$, $g(b)=a$. Then no $G$-invariant tournament
exists on $A$: invariance would give $(a,b)\in O\Leftrightarrow
(b,a)\in O$, contradicting totality and antisymmetry on the pair.
For the session symmetry $\Aut(A,E)$, the transposition of two
same-session values (fixing everything else) is such a $g$; for the
betweenness symmetry $\Aut(\mathbb Q,B)$, the reflection $x\mapsto
a+b-x$ is. Hence neither atom orbit is orientable. (This settles
only the atom orbit; orientability of other atoms of these universes
is not decided here, and the inequivalence Corollary~\ref{cor:noneq}
requires an orientable atom somewhere.)
\end{rem}

\subsection{Generating fresh values and counting freshness sorts}
\label{subsec:freshmodelling}

The step ``pick a fresh name; nothing below depends on which one'' is
the most frequently invoked inference step in nominal reasoning, and it is the one a symbolic implementation must support directly.
Theorem~\ref{thm:someanychar} says that its unrestricted form is
available, among closed data symmetries, exactly for pure names: as soon
as a proper closed symmetry makes additional structure observable, some
finite context admits distinguishable values outside its registers, and
``fresh'' ceases to be a single notion.

What takes its place is stated by Proposition~\ref{prop:someany} and is
already the practice of the field. Over ordered atoms one must generate
a fresh value \emph{in a specified interval} determined by the
registers; over graph atoms, a fresh vertex \emph{with a specified
adjacency pattern} to the registers; over session atoms, a value either
in one of the registered sessions or in a new one. The freshness sorts
of Definition~\ref{def:sorts} are exactly these cases, and
Proposition~\ref{prop:sortcounts} counts them.
At the level of supported unary predicates the sorts are exactly the
atomic cases that a complete symbolic treatment must distinguish, and
this can be said precisely.

\begin{prop}[The algebra of freshness cases]\label{prop:sortalgebra}
Let $G\leq\Sym(A)$ and let $S\subseteq A$ be finite. Write
$\Sigma_S$ for the set of freshness sorts over $S$. Restriction
$\varphi\mapsto\varphi\smallsetminus S$ maps $B^{(1)}_{G,S}$ onto
the Boolean algebra of unions of members of $\Sigma_S$, and the atoms
of that algebra are precisely the sorts in $\Sigma_S$. If $G$ is
locally oligomorphic, then
$m:=|\Sigma_S|=\mathrm{sorts}_G(S)<\infty$ and the algebra has
$2^{m}$ elements.
\end{prop}

\begin{proof}
By Proposition~\ref{prop:someany} the set $\varphi\smallsetminus S$ is a
union of sorts for every $\varphi\in B^{(1)}_{G,S}$; conversely every
union of sorts is $\Fix_G(S)$-invariant, hence is the restriction of a
member of $B^{(1)}_{G,S}$, so the map is onto. It is a Boolean
homomorphism because restriction to the complement of a fixed set
preserves the operations. Its image has the sorts as atoms, since the
sorts are the $\Fix_G(S)$-orbits on $A\smallsetminus S$ and are
therefore nonempty, pairwise disjoint and minimal among unions of
sorts. If $G$ is locally oligomorphic, then $\Sigma_S$ is finite,
since its members are among the finitely many $\Fix_G(S)$-orbits on
$A^1$; the Boolean algebra of all unions of its $m$ atoms consequently
has $2^m$ elements.
\end{proof}

Proposition~\ref{prop:sortcounts} shows that the number of atoms is $1$ for pure names, $k+1$ for the dense order, $2^{k}$ for the random graph, and
$2^{\binom k2}$ for the generic $3$-hypergraph, at a context of size~$k$. For the generic equivalence relation it is $m+1$ when the
context meets $m$ classes, with maximum $k+1$; the betweenness reduct needs
one case for $k\leq1$ and $k+1$ thereafter. Thus the number of atomic
freshness cases grows linearly with the context over ordered data, as
$2^k$ for graph data, and as
$2^{\binom{k}{2}}=2^{\Theta(k^2)}$ for the generic $3$-hypergraph.
These are semantic counts, not implementation-independent lower bounds on
an algorithm. Within each sort, the classical some/any step remains
sound.

\subsection{State spaces: separating finiteness conditions}
\label{subsec:statespaces}

Automata over atoms take orbit-finite sets as state spaces
\cite{BKL14}. Theorem~\ref{thm:orbitfinite} decomposes that hypothesis
into two independent requirements: uniform finiteness of the context
levels, which says that each register assignment sees only finitely
many states, and a bound on support
size, which says that the number of registers required does not grow
without limit and is the resource constraint imposed by a fixed machine.

The two can be separated. The disjoint union of the sets of injective
$n$-tuples over equality atoms is uniform finite with unbounded
supports: every fixed register assignment sees finitely many states,
but no finite bound on the number of registers suffices for the whole
space. Such a space models computations that may allocate arbitrarily
many names while exposing only finitely many states at each fixed
context.

In the bounded case the proof of Theorem~\ref{thm:orbitfinite} gives a
presentation by finite data, which is the form in which an orbit-finite
state space is normally handled.

\begin{cor}[Presentation by context levels]\label{cor:presentation}
Let $G$ be locally oligomorphic with least supports, let $X$ be a
nonempty orbit-finite object, and put
$k=\max\{|\supp(x)|:x\in X\}$. Let
$T_1,\dots,T_r$ represent the finitely many $G$-orbits of subsets of
$A$ of size at most $k$. Then each~$X^{T_j}$ is finite and
\[
   X=\bigcup_{j\leq r}G\cdot X^{T_j},
   \qquad
   |X/G|\ \leq\ \sum_{j\leq r}\bigl|X^{T_j}\bigr|.
\]
Moreover $k$ is the optimal uniform support bound: every state is
supported by at most $k$ atoms, and no smaller uniform bound suffices.
In register terminology, this is the precise semantic sense in which
$k$ registers suffice.
\end{cor}

\begin{proof}
Finiteness of the levels is Theorem~\ref{thm:orbitfinite}(1) and the
displayed cover is established in the proof of
Theorem~\ref{thm:orbitfinite}(2); the orbit bound follows because each
element of a level contributes at most one orbit. Every state is
supported by a set of size at most $k$, and by the definition of $k$
some state has least support of size exactly $k$, so no smaller bound
holds.
\end{proof}

The two numbers $r$ and $\max_j|X^{T_j}|$ are the natural measures of
such a presentation: $r$ counts the shapes of register assignment that
must be distinguished and $|X^{T_j}|$ the states visible from one of
them. Both are finite exactly under the hypotheses of
Theorem~\ref{thm:orbitfinite}, which is the sense in which that theorem
separates a structural condition from a resource constraint. The
corollary is not by itself an effective presentation: computability of
the action and of orbit representatives is additional data. Dropping
the support bound raises the following question.

\begin{qu}[Automata over uniform finite state spaces]
\label{q:unifautomata}
Develop the automata theory of uniform finite state spaces of
unbounded support over a locally oligomorphic symmetry. Which of the
decidability results for orbit-finite automata survive, under what
effective presentation of the levels $X^{S}$, and what is the correct
notion of determinization when the support bound is replaced by a
per-context finiteness condition?
\end{qu}

\subsection{The size of the data reservoir}
\label{subsec:reservoir}

For pure equality atoms, changing the infinite cardinality of the atom
reservoir does not change the abstract action topos: by
Corollary~\ref{thm:regime}, every $\FSS_{\Sym(A)}(A)$ is equivalent to the
Schanuel topos. Consequently statements formulated purely in the internal
categorical language transfer between these presentations. This does not
identify their chosen atom objects, underlying-set functors, or external
cardinalities, and it says nothing about replacing an infinite data domain by
a finite implementation domain. The equivalence is structural, not an
effectiveness theorem. Conversely, Remark~\ref{rem:reconfail} shows that
the bare topos need not determine the presenting symmetry group.

The qualification is material (Theorem~\ref{thm:material}). Cardinality
separates the universes as models of set theory with atoms:
$|\Pfs(A)|=|A|<2^{|A|}$. The failure of choice is nevertheless
\emph{uniform}: the set of unordered pairs of atoms admits no finitely
supported choice function over any infinite $A$, by one and the same
argument. For computation this uniformity matters, because the failure of
choice is not merely a foundational curiosity but a working constraint:
orbit-finitely spanned vector spaces need not admit equivariant
bases~\cite{BFKM24}, while orbit-finite linear programming requires finite
representations adapted to the group action~\cite{GHL23}. The
pair-choice obstruction itself is therefore not an artifact of
countability. Transferring particular algorithmic representations to
uncountable reservoirs would require additional effectiveness
assumptions and is not claimed here.

\subsection{The model-selection criteria}\label{subsec:checklist}

Table~\ref{tab:checklist} summarizes model-selection criteria for the canonical atom object and its invariance hierarchy. The LS and Join columns are structural properties of that hierarchy. The Tournament column indicates an invariant tournament on the canonical atom orbit; for the free-amalgamation rows, non-orientability is stronger and holds for every atom of the universe. Thus the table is a guide to selecting a presented data symmetry, while categorical separation uses the existence of an orientable atom anywhere in the universe.

\begin{table}[htbp]
\centering
\caption{LS denotes least supports, Join the join law, Sorts the maximum
number of freshness sorts among contexts of size $k$
(Proposition~\ref{prop:sortcounts}), and Tournament an invariant
tournament on the canonical atom orbit. In the three free-amalgamation
rows, ``no'' holds for every atom of the universe; for $\Aut(A,E)$ and
$\Aut(\mathbb Q,B)$ it is established only for the canonical atom orbit.
A sort count of $1$ at \emph{every} context occurs only in the first
row, by Theorem~\ref{thm:someanychar}; the last row has one sort for
$k\leq1$ only.}
\label{tab:checklist}
\footnotesize
\begin{tabular}{@{}>{\raggedright\arraybackslash}p{2.55cm}>{\raggedright\arraybackslash}p{2.55cm}>{\centering\arraybackslash}p{0.75cm}>{\centering\arraybackslash}p{0.85cm}>{\centering\arraybackslash}p{2.15cm}>{\centering\arraybackslash}p{2.0cm}@{}}
\toprule
Symmetry & Models & LS & Join & Sorts & Tournament\\
\midrule
$\Sym(A)$ & pure names & \checkmark & \checkmark & $1$ & no\\
$\Aut(\mathbb Q,<)$ & timestamps, priorities & \checkmark & \checkmark & $k+1$ & yes\\
$\Aut(\mathbb G)$ & linked data & \checkmark & \checkmark & $2^{k}$ & no\\
$\Aut(A,E)$ & sessions, shards & $\times$ & \checkmark & $k+1$ & no\\
$\Aut(\mathcal H_3)$ & ternary linked data & \checkmark & $\times$ & $2^{\binom k2}$ & no\\
$\Aut(\mathbb Q,B)$ & unoriented order data & $\times$ & $\times$ & $1$ ($k\leq1$); $k+1$ ($k\geq2$) & no\\
\bottomrule
\end{tabular}
\end{table}

LS decides whether canonical register contents and least-support reasoning are available; Join decides whether invariance decomposes across unions of register sets; Sorts counts the kinds of fresh value an implementation must be able to generate, and equals~$1$ at every context only in the nominal setting; and a tournament on the canonical atom orbit certifies categorical separation from every nominal universe. Absence of such a tournament is inconclusive and motivates the finer spectrum of Question~\ref{q:spectrum}. The atom cardinality appears in no column: it affects the material presentation, not this structural profile. The table therefore separates two modelling decisions: choose $|A|$ for the intended scale of the data domain, and choose $G$ for the structure programs may observe.

\section{Open problems and boundaries}\label{sec:problems}

The results highlight some gaps in the current structural framework.

\begin{qu}[Meet law beyond trivial algebraic closure]
\label{q:charmeet}
For closed oligomorphic $G=\Aut(M)$ with degenerate algebraic closure,
Proposition~\ref{prop:weakEI} identifies the meet law with weak
elimination of imaginaries. Develop the corresponding theory when
$\operatorname{acl}$ is nontrivial. In an $\omega$-categorical
structure, $\operatorname{acl}(S)$ is finite for finite $S$; one
natural possibility is therefore to index contexts by finite
algebraically closed sets. Proposition~\ref{prop:bklpathologies}(a), the diagonal symmetry of \cite[Example~9.1]{BKL14}, provides a concrete closed oligomorphic test case with nontrivial stabilizer-definable (hence algebraic) closure in which the ordinary meet law already fails at arity~$1$. Determine whether indexing by finite algebraically closed contexts restores an exact
lattice and stabilizer criterion. In the same framework, characterize
the join law in model-theoretic terms, especially for ternary
ultrahomogeneous structures.
\end{qu}

\begin{qu}[Pointed reconstruction and uncountable relation theory]
\label{q:reconstruction}
For closed oligomorphic $G,H\leq\Sym(\omega)$, does an equivalence
$\FSS_G(\omega)\simeq\FSS_H(\omega)$ of the bare atomic toposes
necessarily transport the canonical point and hence induce a
topological isomorphism $G\cong H$?  For countable
$\omega$-categorical structures, the corresponding pointed
reconstruction is governed by Ahlbrandt--Ziegler bi-interpretability
\cite{AZ86}. The point is essential in the
representation theory of atomic and Galois topoi
\cite{Dubuc04}; coarse-group and weak-elimination methods give
substantial reconstruction information on the group side
\cite{Paolini24}, but do not by themselves show
that an arbitrary topos equivalence preserves the point.

A related cardinality-free problem is to characterize intrinsically the
families of finitary relation algebras that arise as
$\Inv_\bullet(G)$ for closed $G\leq\Sym(A)$ when $A$ is
uncountable. Such a descriptive theory would clarify which parts of
small-index and reconstruction methods survive beyond the Polish
setting and would give an intrinsic account of the invariance hierarchy
at every atom cardinality.
\end{qu}

\begin{qu}[Invariant-structure spectrum]\label{q:spectrum}
Orientability is the first entry in a hierarchy of equivalence
invariants of $\FSS_G(A)$. For a finite relational signature
$\lambda$, ask whether an atom carries an invariant
$\lambda$-structure of a prescribed isomorphism type; this can be
expressed through subobject lattices of finite powers. Develop this
spectrum, compute it for standard homogeneous data symmetries, and
determine how much of the symmetry parameter it reconstructs. A first
target is to separate
$\FSS_{\Aut(\mathbb G)}(\mathbb G)$, for the random graph
$\mathbb G$, from $\Nom$ by an invariant on atoms and their finite
powers.
\end{qu}

\begin{qu}[Effective and quantitative context composition]\label{q:effectivejoin}
Suppose a locally oligomorphic data symmetry is given together with an
effective finite representation of stabilizer orbits. When can the
$S\cup T$-level of the invariance hierarchy be computed compositionally
from the $S$- and $T$-levels, and what is the complexity of deciding such
decomposability for a represented supported relation?  Theorem~\ref{thm:join}
settles the structural equality for binary ultrahomogeneous languages,
whereas Theorems~\ref{thm:hypergraph} and~\ref{thm:betwfail} show that
higher-arity or reduct structure can create genuinely mixed information.
Develop quantitative measures of that mixed-context information and relate
them to the size and complexity of symbolic representations used by
automata over structured data.
\end{qu}

\section{Conclusion}\label{sec:conclusion}

The finite-dependence profile refines the yes/no support condition by
indicating what becomes observable at each finite context. For a finitely supported $G$-set this profile is $S\mapsto X^S$, and for relations on the
atom object it becomes the invariance hierarchy
$S\mapsto B^{(n)}_{G,S}$. Its value is not that it collapses nominal
principles to one condition, but that it places distinct principles in a
common language.

The meet law is the relational form of the established BKL
least-support criterion, whereas level injectivity is
fungibility. The join law is independent and measures whether observations
over a union of contexts are generated by observations over its parts.
Freshness remains uniform within stabilizer orbits, but its unrestricted
some/any form characterizes full equality symmetry among closed groups. For
symmetries with least supports, orbit-finiteness implies finite visibility at
every context and a uniform support bound; under local oligomorphicity these
conditions are sufficient as well. Contextual orientability supplies a
different, pointed kind of information, and its threshold is exactly the
least support size of a choice function on unordered pairs.

The hierarchy also clarifies what categorical semantics can forget.
Replacing the acting group by its pointwise closure does not change the
finite-support action category, and suitable ultrahomogeneous presentations
with the same age yield equivalent action topoi. Such equivalences need not
preserve a chosen atom object or its external cardinality. Material FSS
universes retain that information, so categorical and material comparisons
answer different questions.

These results do not constitute a reconstruction theorem. Some invariants
are pointed, the present tests do not separate all standard homogeneous data
symmetries, and the model-theoretic characterization of the meet law still
requires countable $\omega$-categoricity and degenerate algebraic closure.
The next problems are therefore concrete: identify the appropriate contexts
when algebraic closure is nontrivial, quantify mixed-context information
beyond the binary join theorem, and develop categorical invariants finer than orientability.

\bibliographystyle{alphaurl}

\begin{thebibliography}{BFKM24}

\bibitem[AZ86]{AZ86}
Gisela Ahlbrandt and Martin Ziegler.
\newblock Quasi finitely axiomatizable totally categorical theories.
\newblock {\em Annals of Pure and Applied Logic}, 30(1):63--82, 1986.
\newblock \href {https://doi.org/10.1016/0168-0072(86)90037-0}
  {\path{doi:10.1016/0168-0072(86)90037-0}}.

\bibitem[AC20]{AC20}
Andrei Alexandru and Gabriel Ciobanu.
\newblock {\em Foundations of Finitely Supported Structures: A Set Theoretical
  Viewpoint}.
\newblock Springer, 2020.
\newblock \href {https://doi.org/10.1007/978-3-030-52962-8}
  {\path{doi:10.1007/978-3-030-52962-8}}.

\bibitem[AC22]{AC22}
Andrei Alexandru and Gabriel Ciobanu.
\newblock Various forms of infinity for finitely supported structures.
\newblock {\em Archive for Mathematical Logic}, 61:173--222, 2022.
\newblock \href {https://doi.org/10.1007/s00153-021-00787-2}
  {\path{doi:10.1007/s00153-021-00787-2}}.

\bibitem[BFKM24]{BFKM24}
Miko{\l}aj Boja\'nczyk, Joanna Fijalkow, Bartek Klin, and Joshua Moerman. Orbit-finite-dimensional vector spaces and weighted register
 automata. {\em TheoretiCS}, 3:article 13, 1--41, 2024.
\newblock \href {https://doi.org/10.46298/theoretics.24.13}
  {\path{doi:10.46298/theoretics.24.13}}.

\bibitem[BKL14]{BKL14}
Miko{\l}aj Boja\'nczyk, Bartek Klin, and S{\l}awomir Lasota.
\newblock Automata theory in nominal sets.
\newblock {\em Logical Methods in Computer Science}, 10(3):1--44, 2014.
\newblock \href {https://doi.org/10.2168/LMCS-10(3:4)2014}
  {\path{doi:10.2168/LMCS-10(3:4)2014}}.

\bibitem[Cam76]{Cameron76}
Peter~J. Cameron.
\newblock Transitivity of permutation groups on unordered sets.
\newblock {\em Mathematische Zeitschrift}, 148:127--139, 1976.
\newblock \href {https://doi.org/10.1007/BF01214702}
  {\path{doi:10.1007/BF01214702}}.

\bibitem[Cam99]{Cameron99}
Peter~J. Cameron.
\newblock {\em Permutation Groups}, volume~45 of {\em London Mathematical
  Society Student Texts}.
\newblock Cambridge University Press, 1999.
\newblock \href {https://doi.org/10.1017/CBO9780511623677}
  {\path{doi:10.1017/CBO9780511623677}}.

\bibitem[Car16]{Caramello16}
Olivia Caramello.
\newblock Topological {G}alois theory.
\newblock {\em Advances in Mathematics}, 291:646--695, 2016.
\newblock \href {https://doi.org/10.1016/j.aim.2015.11.050}
  {\path{doi:10.1016/j.aim.2015.11.050}}.

\bibitem[CF04]{CasanovasFarre}
Enrique Casanovas and Rafel Farr\'e.
\newblock Weak forms of elimination of imaginaries.
\newblock {\em Mathematical Logic Quarterly}, 50(2):126--140, 2004.
\newblock \href {https://doi.org/10.1002/malq.200310083}
  {\path{doi:10.1002/malq.200310083}}.

\bibitem[Con17]{Conant17}
Gabriel Conant.
\newblock An axiomatic approach to free amalgamation.
\newblock {\em The Journal of Symbolic Logic}, 82(2):648--671, 2017.
\newblock \href {https://doi.org/10.1017/jsl.2016.42}
  {\path{doi:10.1017/jsl.2016.42}}.

\bibitem[Dub04]{Dubuc04}
Eduardo~J. Dubuc.
\newblock On the representation theory of {G}alois and atomic topoi.
\newblock {\em Journal of Pure and Applied Algebra}, 186:233--275, 2004.
\newblock \href {https://doi.org/10.1016/S0022-4049(03)00141-5}
  {\path{doi:10.1016/S0022-4049(03)00141-5}}.

\bibitem[GHL25]{GHL23}
Arka Ghosh, Piotr Hofman, and S{\l}awomir Lasota.
\newblock Orbit-finite linear programming.
\newblock {\em Journal of the ACM}, 72(1):1:1--1:39, 2025.
\newblock Conference version: Proc.\ 38th Annual ACM/IEEE Symposium on Logic in
  Computer Science (LICS 2023).
\newblock \href {https://doi.org/10.1145/3703909} {\path{doi:10.1145/3703909}}.

\bibitem[Hod93]{Hodges93}
Wilfrid Hodges.
\newblock {\em Model Theory}.
\newblock Cambridge University Press, 1993.

\bibitem[Joh02]{Johnstone02}
Peter~T. Johnstone.
\newblock {\em Sketches of an Elephant: A Topos Theory Compendium}.
\newblock Oxford University Press, 2002.

\bibitem[JZ08]{JunkerZiegler08}
Markus Junker and Martin Ziegler.
\newblock The 116 reducts of $(\mathbb{Q},<,a)$.
\newblock {\em The Journal of Symbolic Logic}, 73(3):861--884, 2008.
\newblock \href {https://doi.org/10.2178/jsl/1230396752}
  {\path{doi:10.2178/jsl/1230396752}}.

\bibitem[Mac11]{Macpherson11}
Dugald Macpherson.
\newblock A survey of homogeneous structures.
\newblock {\em Discrete Mathematics}, 311(15):1599--1634, 2011.
\newblock \href {https://doi.org/10.1016/j.disc.2011.01.024}
  {\path{doi:10.1016/j.disc.2011.01.024}}.

\bibitem[MLM94]{MacLaneMoerdijk}
Saunders Mac~Lane and Ieke Moerdijk.
\newblock {\em Sheaves in Geometry and Logic: A First Introduction to Topos
  Theory}.
\newblock Springer, 1994.

\bibitem[MT11]{MacphersonTent11}
Dugald Macpherson and Katrin Tent.
\newblock Simplicity of some automorphism groups.
\newblock {\em Journal of Algebra}, 342:40--52, 2011.
\newblock \href {https://doi.org/10.1016/j.jalgebra.2011.05.021}
  {\path{doi:10.1016/j.jalgebra.2011.05.021}}.

\bibitem[Pao24]{Paolini24}
Gianluca Paolini.
\newblock The isomorphism problem for oligomorphic groups with weak elimination
  of imaginaries.
\newblock {\em Bulletin of the London Mathematical Society}, 56(8):2597--2614,
  2024.
\newblock \href {https://doi.org/10.1112/blms.13086}
  {\path{doi:10.1112/blms.13086}}.

\bibitem[Pit13]{Pitts13}
Andrew~M. Pitts.
\newblock {\em Nominal Sets: Names and Symmetry in Computer Science}.
\newblock Cambridge University Press, 2013.

\end{thebibliography}

\end{document}